\documentclass[11pt]{article}

\usepackage{amsthm}
\usepackage{amsmath}
\usepackage{amssymb}
\usepackage{graphicx}
\usepackage[round]{natbib}
\usepackage{bbm}
\usepackage{textcomp}
\usepackage{url}
\usepackage[margin=0.9in]{geometry}  
\usepackage{subcaption}
\usepackage{mathtools}
\usepackage{mathrsfs}
\usepackage{multirow}
\usepackage{makecell}
\usepackage{arydshln}
\usepackage{booktabs}
\usepackage{extarrows}
\usepackage{siunitx}
\usepackage{comment}
\usepackage{enumitem}
\usepackage{stmaryrd,scalerel}
\usepackage{algorithm}
\usepackage{algorithmic}

\usepackage{hyperref}
\hypersetup{colorlinks=true,linkcolor=blue,filecolor=blue,citecolor=blue}

\setlist[itemize]{itemsep=4pt, topsep=4pt, parsep=0pt, partopsep=0pt}
\setlist[enumerate]{itemsep=4pt, topsep=4pt, parsep=0pt, partopsep=0pt}

\theoremstyle{plain}
\newtheorem{thm}{Theorem}
\newtheorem{lemma}[thm]{Lemma}

\theoremstyle{definition}

\theoremstyle{remark}

\DeclareMathOperator{\E}{\mathbb{E}}
\DeclareMathOperator{\Var}{Var}

\DeclareMathOperator{\tr}{tr}

\DeclareMathOperator{\vecl}{vecl}

\newcommand{\bvar}{\boldsymbol{\varepsilon}}
\newcommand{\bnu}{\boldsymbol{\nu}}
\newcommand{\bpi}{\boldsymbol{\pi}}
\newcommand{\bgamma}{\boldsymbol{\gamma}}
\newcommand{\brho}{\boldsymbol{\rho}}
\newcommand{\bbeta}{\boldsymbol{\beta}}
\newcommand{\balpha}{\boldsymbol{\alpha}}
\newcommand{\bmu}{\boldsymbol{\mu}}
\newcommand{\boldeta}{\boldsymbol{\eta}}
\newcommand{\bOmega}{\boldsymbol{\Omega}}
\newcommand{\bSigma}{\boldsymbol{\Sigma}}
\newcommand{\bLambda}{\boldsymbol{\Lambda}}
\newcommand{\bXi}{\boldsymbol{\Xi}}

\newcommand{\ba}{{\mathbf a}}
\newcommand{\be}{{\mathbf e}}

\newcommand{\bx}{{\mathbf x}}
\newcommand{\by}{{\mathbf y}}

\newcommand{\bw}{{\mathbf w}}

\newcommand{\bD}{{\bf D}}
\newcommand{\bA}{{\bf A}}
\newcommand{\bB}{{\bf B}}

\newcommand{\bE}{{\bf E}}
\newcommand{\bG}{{\bf G}}
\newcommand{\bI}{{\bf I}}
\newcommand{\bK}{{\bf K}}

\newcommand{\bP}{{\bf P}}
\newcommand{\bS}{{\bf S}}
\newcommand{\bX}{{\bf X}}

\newcommand{\bR}{{\bf R}}

\newcommand{\bV}{{\bf V}}
\newcommand{\bQ}{{\bf Q}}
\newcommand{\bJ}{{\bf J}}
\newcommand{\bW}{{\bf W}}
\newcommand{\bM}{{\bf M}}
\newcommand{\bT}{{\bf T}}
\newcommand{\bH}{{\bf H}}

\begin{document}

\def\spacingset#1{\renewcommand{\baselinestretch}%
{#1}\small\normalsize} \spacingset{1}


\title{\bf Generalized Correlation Regression for Disentangling Dependence in Clustered Data}
\author{
    Yibo Wang\footnote{Department of Statistics and Finance, University of Science and Technology of China. Email: yibooo@mail.ustc.edu.cn. This research was conducted during a visit with Professor Leng, supported by a grant from the University of Science and Technology of China},~~
    Chenlei Leng\footnote{Department of Applied Mathematics, Hong Kong Polytechnic University. Email: chenlei.leng@gmail.com},~~
    Cheng Yong Tang\footnote{Department of Statistics, Operations, and Data Science, Temple University. Email: yongtang@temple.edu}
    \medskip
    \\
    University of Science and Technology of China$^*$,\\
    Hong Kong Polytechnic University$^\dag$, \\
    Temple University$^\ddag$
}
\date{}
\maketitle

\bigskip

\begin{abstract}
Clustered and longitudinal data are pervasive in scientific studies, from prenatal health programs to clinical trials and public health surveillance. Such data often involve non-Gaussian responses---including binary, categorical, and count outcomes---that exhibit complex correlation structures driven by multilevel clustering, covariates, over-dispersion, or zero inflation. Conventional approaches such as generalized mixed-effects models (GLMMs) and generalized estimating equations (GEEs) can capture some of these dependencies, but they are often too rigid or impose restrictive assumptions that limit interpretability and predictive performance. 

We investigate \emph{generalized correlation regression} (GCR), a unified framework that models correlations directly as functions of interpretable covariates while simultaneously estimating marginal means. By applying a generalized $z$-transformation, GCR guarantees valid correlation matrices, accommodates unbalanced cluster sizes, and flexibly incorporates covariates such as time, space, or group membership into the dependence structure. Through applications to modern prenatal care, a longitudinal toenail infection trial, and clustered health count data, we show that GCR not only achieves superior predictive performance over standard methods, but also reveals family-, community-, and individual-level drivers of dependence that are obscured under conventional modeling. These results demonstrate the broad applied value of GCR for analyzing binary, count, and categorical data in clustered and longitudinal settings.
\end{abstract}

\noindent%
\noindent\textit{Keywords:} Clustered data; Longitudinal data; Correlation regression; Generalized $z$-transformation; Regression modeling
\vfill

\newpage
\spacingset{1} 

\section{Introduction}\label{sec:introduction}

\subsection{Background and Overview}

Dependent data are ubiquitous in applied statistics. Two common settings illustrate this:  
(i) sequential observations, leading to time-series or longitudinal measurements; and  
(ii) clustered or multilevel observations, such as patients within hospitals or children within classrooms.  
In both cases, correlation is intrinsic: repeated measurements on the same individual tend to move together over time, and individuals sharing the same setting are subject to common influences---protocols, environments, or practitioner effects.  

Failure to account for such dependence can lead to biased effect estimates, miscalibrated standard errors, efficiency loss, and ultimately unreliable conclusions. A central challenge in applied statistics, therefore, is to model dependence in ways that are both flexible and interpretable, while remaining computationally feasible. This requires methods that handle heterogeneous cluster sizes and irregular timing, allow correlation to vary with covariates, space, or time, and deliver valid inference even under partial model misspecification. Balancing practical relevance with methodological rigor motivates this paper.  

For continuous outcomes, a rich toolbox already exists. Classical linear mixed-effects models represent correlations through latent random effects \citep{laird1982random, pinheiro2009mixed}; covariance pattern models directly impose structured correlation matrices \citep{stroup2018sas}; and unconstrained parameterizations such as the modified Cholesky decomposition or matrix logarithmic transforms enable flexible yet interpretable modeling of high-dimensional covariance structures \citep{pourahmadi1999joint, zhang2015joint}.  

For broader outcome types such as binary, categorical, and count data, the challenges are considerably greater.  
Unlike Gaussian outcomes, where covariance modeling provides a unified framework, discrete responses impose intrinsic constraints.  
Binary and categorical variables are bounded by their outcome space (e.g., 0/1 or a finite set of categories), which severely restricts the allowable dependence structures.  
Count data often display over-dispersion or zero inflation, further complicating the specification of both marginal distributions and correlations.  
Such features arise routinely in applied studies---clustered health outcomes, longitudinal educational assessments, or ecological surveys---where complex dependence is intertwined with distributional irregularities.  
Consequently, it is effectively infeasible to construct a single joint distributional model that simultaneously respects the marginal constraints of non-Gaussian outcomes and captures their within-cluster correlations.  
This limitation stands in sharp contrast to the Gaussian case, where matrix-based formulations provide a coherent and tractable strategy that has been widely adopted in practice.  


As detailed in Section~\ref{sec:real_1}, our primary illustrating example is the 
\emph{modern prenatal care dataset}, a landmark study conducted by researchers at Princeton University 
\citep{Pebley:1996, Rodriguez:2001}. 
The dataset records clustered binary outcomes on whether births received modern prenatal care, drawn from the 
\emph{1987 National Survey of Maternal and Child Health in Guatemala}. 
After preprocessing, it consists of 2,223 births from 1,414 mothers across 160 communities, with covariates measured at three levels: 
community-level (e.g., accessibility of healthcare facilities), 
maternal-level (e.g., ethnic-linguistic identity and socioeconomic status), 
and individual-level (e.g., maternal age). 
This creates a natural \emph{multi-level clustering structure}, where correlations arise both within mothers and within communities.

This dataset raises several critical scientific questions of both methodological and applied importance:  
\begin{enumerate}
    \item \textbf{Covariate effects:} How do community-, maternal- and individual-level covariates shape prenatal care utilization?
    \item \textbf{Dependence beyond hierarchy:} Are there systematic correlations in care choices across mothers with similar socioeconomic or geographic backgrounds?  
    \item \textbf{Equity and intervention:} How can identifying such dependence patterns highlight sources of inequity and inform targeted, cost-effective public health interventions?
\end{enumerate}

Answering these questions is of profound practical importance. 
Prenatal care utilization is widely recognized as a key determinant of maternal and infant health outcomes, and identifying the drivers of correlation---whether arising from community, maternal, or cross-level factors---has direct implications for targeting interventions. 
Detecting such dependence not only clarifies sources of inequity but also informs the design of cost-effective policies and field studies aimed at improving maternal and child health in resource-constrained settings.

The binary response indicates whether a given birth received modern prenatal care.  
Existing analyses typically model such choices using community-, maternal-, and individual-level covariates, but they often ignore or only partially capture the dependence inherent in the data.  
In reality, responses are not independent: births from the same mother and mothers from the same community share unobserved influences that induce correlation.  
Our analysis further reveals that dependence extends beyond hierarchical clustering, with correlations also arising across mothers who share \emph{similar covariates}---for example, those in comparable socioeconomic contexts.  
Such patterns indicate potential systematic inequities in healthcare access that conventional models overlook.  
From a substantive perspective, detecting these cross-cluster dependencies is essential: it confirms the role of observed covariates, reveals hidden sources of inequity, and identifies population subgroups where targeted field studies and interventions may be most effective.  

This example highlights three central \emph{statistical challenges}:  
\begin{enumerate}
    \item \textbf{Complex within-cluster dependence}, arising from multiple hierarchical levels (births nested within mothers, and mothers nested within communities);  
    \item \textbf{Non-Gaussian outcomes}, where binary responses preclude direct use of Gaussian-based correlation modeling;  
    \item The need for \textbf{flexible and interpretable correlation structures} that capture not only traditional clustering effects but also dependence driven by covariates.  
\end{enumerate}

Motivated by these challenges, we propose an approach that can be informally summarized as a \emph{regression model for correlations}:  
\[
\text{(Informal)} \qquad \text{Correlation matrix } \bR \;\sim\; \text{Covariates},
\]
where the covariates may include both continuous and categorical factors, drawn from multiple hierarchical levels.  
Unlike conventional regression, which links covariates to the \emph{mean response}, our framework links them directly to the \emph{dependence structure}, enabling a flexible yet interpretable characterization of within- and cross-cluster correlations.

The modeling device is key to answering the scientific questions raised by our motivating datasets. By treating correlations as regression targets, it uncovers dependence patterns beyond mean-based analyses. In the prenatal care study, for instance, it enables formal testing of whether correlations in healthcare choices arise not only within mothers or communities but also across mothers with similar socioeconomic profiles, thereby uncovering potential inequities that conventional models fail to detect, highlighting questions worth further field study or experimental investigation.

\subsection{Existing Methods}

Classical approaches to clustered and longitudinal data rely on linear mixed-effects models \citep{laird1982random,pinheiro2009mixed} and covariance pattern models \citep{stroup2018sas}, which accommodate structured correlations via random effects or parametric covariance specifications. Extensions to non-Gaussian outcomes through generalized linear mixed models (GLMMs) are well established \citep{breslow1993approximate}, but they often require restrictive random-effects structures, and inference can be computationally burdensome, with likelihood-based tests exhibiting convergence and asymptotic challenges \citep{Self:1987, mcculloch2004generalized}.  

The generalized estimating equations (GEE) framework \citep{liang1986longitudinal} provides a popular marginal alternative, leveraging a ``working" correlation structure to improve efficiency in mean parameter estimation. Yet, commonly employed specifications---exchangeable, autoregressive, or independence---rarely capture the complex multi-level dependence encountered in modern applications. Recent extensions have sought to enrich this framework by modeling covariance as a function of covariates, for example through covariance regression \citep{pourahmadi1999joint, pourahmadi2000maximum, hoff2012covariance, fox2015bayesian, zou2017covariance}. These methods open the door to regression-like formulations of correlation, but most advances remain restricted to Gaussian outcomes.  

Copula-based methods \citep{nelsen2006introduction,peter2007correlated,nikoloulopoulos2013copula} have gained traction as a flexible means of decoupling marginal distributions from dependence, and recent developments have adapted copulas to clustered binary and ordinal responses \citep{panagiotelis2012pair, krupskii2025factor}. Despite their flexibility, copula models often lack direct interpretability in terms of regression effects on correlations, and computational scalability becomes problematic in higher-dimensional or multi-level clustering scenarios.  

Closest to our work is the recent applied formulation of correlation regression in \citet{hu2024applied}, which directly regresses correlations on covariates. While conceptually appealing, their method is confined to Gaussian responses and does not naturally extend to binary, categorical, or count data, thereby leaving a gap that our work addresses.

\subsection{Our Contributions}  

We introduce a \emph{generalized correlation regression (GCR)} framework that directly treats correlations as regression targets, providing a unified approach for both Gaussian and non-Gaussian clustered data. The core methodological innovation is the \emph{generalized \(z\)-transformation}: by applying a matrix logarithm to the correlation matrix and vectorizing its strictly lower-triangular elements, correlations are mapped to an unconstrained Euclidean space. This transformation guarantees positive definiteness upon inversion, naturally accommodates unbalanced cluster sizes, and allows pairwise correlations to depend on interpretable covariates such as spatial or temporal proximity, group membership, or baseline similarity. Estimation proceeds via an estimating-equation and pseudo-likelihood strategy, ensuring numerical stability, accuracy, and practical implementability.  

Through multiple real-world applications---including modern prenatal care, a toenail infection clinical trial, and clustered health count data---we demonstrate how GCR enables fine-grained exploration of dependence beyond standard mixed-effects or GEE models. We provide a practical guide for constructing covariates, visualizing cluster-specific dependence, and interpreting parameter estimates. For formal inference, the framework supports benchmark tests, such as large-sample Wald tests, allowing rigorous assessment of whether observed dependence patterns differ significantly from simpler null structures (e.g., independence or exchangeable correlation). These capabilities are particularly valuable in applied studies, where uncovering residual dependence informs targeted interventions and experimental design.  

Methodologically, GCR contributes a flexible, interpretable, and valid parameterization of correlation matrices for a wide range of response types, including binary, categorical, and count outcomes. We establish theoretical guarantees for identifiability, consistency, and asymptotic behavior of the proposed estimators. Moreover, our approach provides a practical solution for testing random or clustering effects in classical mixed-effects models---tasks that are notoriously difficult due to boundary constraints and computational intractability---by leveraging the generalized \(z\)-transformation for stable, valid inference.  

Taken together, GCR offers a comprehensive, application-oriented framework for modeling and inferring complex correlation structures in clustered and longitudinal data, combining interpretability, flexibility, and rigorous statistical support.

The paper is organized as follows. Section~\ref{sec:methodology} presents the proposed methodology, including model specification, the generalized $z$-transformation, and parameter estimation. Section~\ref{sec:real_1} illustrates the approach using the prenatal care dataset. Section~\ref{sec:real_2} provides an additional example on a longitudinal binary dataset. Additional technical details, simulations, and further data analysis are in the Supplementary Material. All code is publicly available on \url{https://github.com/yibowang2004/gcr} for reproducibility and practical use. All GEE and GLMM analyses in this paper are conducted using the \texttt{gee} package \citep{Carey:2024} and the \texttt{glmmTMB} package \citep{McGillycuddy:2025} in \textsf{R}, respectively.

\section{Methodology}\label{sec:methodology}

\subsection{Model Settings}

Consider a generic grouped response vector 
\(\by_i = (y_{i1}, \dots, y_{im_i})^\top\), 
where \(y_{ij}\) denotes the \(j\)th observation in group \(i\), and the group size \(m_i\) may vary. Associated covariates are collected in 
\(\bX_i = (\bx_{i1}, \dots, \bx_{im_i})^\top\), 
with each \(\bx_{ij} \in \mathbb{R}^p\) representing characteristics linked to \(y_{ij}\). This framework naturally encompasses diverse applied settings. For instance, in a prenatal care study, \(i\) indexes communities, \(y_{ij}\) indicates whether the \(j\)th birth received modern care, and \(\bx_{ij}\) records maternal and household covariates; in an ecological study, \(i\) may index nesting locations, \(y_{ij}\) could be chick tick counts, and \(\bx_{ij}\) includes brood height or year of observation. This setup allows explicit modeling of within-group correlations while accommodating individual-level covariates.

For the response \(y_{ij}\), we assume a generalized linear model (GLM) with density
\begin{equation}\label{eq:marginal}
    f(y_{ij}) = \exp\!\Big\{\frac{y_{ij}\theta_{ij} - a(\theta_{ij}) + b(y_{ij})}{\phi}\Big\},
\end{equation}
where the canonical parameter \(\theta_{ij} = h(\bx_{ij}^\top \bbeta)\) for a known link function \(h(\cdot)\).  
Here, \(\bbeta\) denotes the regression coefficients of primary interest, and \(\phi\) is a nuisance scale parameter that may be known (e.g., for Bernoulli or Poisson responses) or estimated from the data.

Assuming the marginal distribution of \(y_{ij}\) follows \eqref{eq:marginal}, its first two moments are
\begin{align}\label{eq:moments}
    \E(y_{ij}) = a'(\theta_{ij}), 
    \qquad 
    \Var(y_{ij}) = \phi\, a''(\theta_{ij}),
\end{align}
where \(a'(\cdot)\) denotes the first derivative of \(a(\cdot)\).  
Since the GLM sets \(\theta_{ij} = h(\bx_{ij}^\top \bbeta)\), the mean function 
\(\mu(\cdot) = a'\big(h(\cdot)\big)\) links the linear predictor \(\bx_{ij}^\top \bbeta\) directly to \(\E(y_{ij})\).  
In practice, this provides a clear interpretation: once covariates \(\bx_{ij}\) are specified---for example, treatment and time in a clinical trial or chick characteristics in an ecological study---the GLM framework systematically quantifies how these predictors influence the expected outcome.

Several common distributions arise as special cases within the GLM framework:  
\begin{itemize}
    \item[1.] \textbf{Gaussian (continuous data):}  
    $a'(x) = x$, \ $h(x) = x \;\Rightarrow\; \mu(x) = x$.  

    \item[2.] \textbf{Poisson (count data):}  
    $a'(x) = e^x$, \ $h(x) = x \;\Rightarrow\; \mu(x) = e^x$.  

    \item[3.] \textbf{Bernoulli (binary data):}  
    $a'(x) = \dfrac{e^x}{1 + e^x}$, \ $h(x) = x \;\Rightarrow\; \mu(x) = \dfrac{e^x}{1 + e^x}$.  

    \item[4.] \textbf{Gamma (positively skewed continuous data):}  
    $a'(x) = -1/x$, \ $h(x) = -e^{-x} \;\Rightarrow\; \mu(x) = e^x$.  
\end{itemize}

For general responses, specifying the full joint distribution of $\by_i$ is often intractable, though the components $y_{ij}$ are dependent. We therefore adopt a marginal specification via \eqref{eq:marginal}, ensuring the moment conditions in \eqref{eq:moments} hold. In this setting, generalized estimating equations (GEE) \citep{liang1986longitudinal} provide a practical approach for consistently estimating $\bbeta$, and thus the mean and variance functions, without requiring a full joint distribution.

The GEE approach is based on the estimating equations derived from the zero-mean vector $\by_i - \bmu_i$, where $\bmu_i = \E(\by_i)$. A key strength of GEE lies in its explicit use of the covariance structure of $\by_i$, given by  
\begin{align}\label{eq:vcov}
    \bV_i = \bA_i^{1/2} \bR_i \bA_i^{1/2},
\end{align}
where $\bA_i$ is a diagonal matrix with entries $\Var(y_{ij}) = \phi\, a''(\theta_{ij})$, and $\bA_i^{1/2}$ denotes its matrix square root.  

Let $\bD_i = \partial \bmu_i / \partial \bbeta$. The GEE estimator $\hat{\bbeta}_n$ solves
\begin{equation}\label{eq:GEE}
    \sum_{i=1}^n \bD_i^\top \bV_i^{-1} (\by_i - \bmu_i) = \mathbf{0}.
\end{equation}
Under standard regularity conditions, $\hat{\bbeta}_n$ is consistent for $\bbeta_0$ and asymptotically normal, providing a basis for statistical inference; see the Supplementary Material for details.

When the scale parameter $\phi_0$ is unknown, it can be estimated via Pearson residuals:
\begin{equation}\label{eq:phi_1}
    \hat{r}_{ij} = 
    \frac{y_{ij} - a'\!\big(\hat{\theta}_{ij}\big)}
         {\sqrt{a''(\hat{\theta}_{ij})}}, \quad 
    \hat{\theta}_{ij} = h(\bx_{ij}^\top \hat{\bbeta}_n),
\end{equation}
with the corresponding estimator
\begin{equation}\label{eq:phi_2}
    \hat{\phi}_n = \frac{1}{N - p} 
    \sum_{i=1}^n \sum_{j=1}^{m_i} \hat{r}_{ij}^2,
\end{equation}
where $N = \sum_{i=1}^n m_i$ is the total sample size and $p$ the number of regression parameters. 
Under mild conditions, $\hat{\phi}_n$ is consistent for $\phi_0$ and satisfies $\hat{\phi}_n - \phi_0 = O_p(n^{-1/2})$ when $\bbeta_0$ is known \citep{liang1986longitudinal}.

While GEE provides consistent estimates for the marginal mean parameters, its traditional focus is on $\bbeta$, treating the working correlation $\bR_i$ as a nuisance. 
In many applied settings, however, understanding the structure of $\bR_i$ itself is scientifically important: it encodes how outcomes co-vary within clusters and how dependence may vary with covariates. 
This observation motivates our \emph{generalized correlation regression} approach, where $\bR_i$ is directly modeled as a function of interpretable covariates, enabling richer insights into clustering, temporal, and covariate-driven dependence patterns that conventional GEE or GLMM methods may overlook.

\subsection{Our Approach with Generalized z-Transformation}

In the GEE framework, the matrix $\bR_i$ in \eqref{eq:vcov} represents the \emph{working correlation matrix}, capturing correlations among observations within the same group.  
Traditionally, $\bR_i$ is parameterized by a low-dimensional vector $\balpha$ to accommodate a limited set of structures such as exchangeable, autoregressive (AR), moving average (MA), or ARMA-type correlations \citep{liang1986longitudinal}.  
Estimation of $\balpha$ typically relies on moment-based methods or pseudo-likelihood approaches.

However, modeling $\bR_i$ poses several challenges.  
First, each $\bR_i$ must be a valid correlation matrix: symmetric, positive definite, and with unit diagonals.  
Second, the dimension of $\bR_i$ often varies across groups because group sizes $m_i$ are unbalanced.  
Third, correlations may depend on a range of covariates, both continuous and categorical, reflecting individual-level, group-level, or temporal effects.  
In practice, these challenges are ubiquitous: longitudinal patients often have different numbers of visits, and survey respondents may differ in demographic or geographic characteristics.  
While GEE provides a framework for handling correlated data, specifying, estimating, and validating complex correlation structures remains difficult, leading practitioners to rely on simple working assumptions such as independence or exchangeable correlations.

To address these challenges, we adopt the \emph{generalized $z$-transformation}, a flexible device that enables parsimonious modeling of $\bR_i$ while automatically ensuring that it remains a valid correlation matrix.  
This approach is particularly advantageous in applied settings with heterogeneous group sizes, providing a coherent framework without ad hoc fixes or restrictive assumptions.  
Moreover, the parameters obtained under the generalized $z$-transformation retain intuitive interpretations, linking correlation strength to covariates or structural features of the data (e.g., temporal ordering or group membership), thereby enhancing transparency and practical usability for researchers and practitioners.

We begin by introducing some notation.  
Let $\mathbf{A} \in \mathbb{R}^{m \times m}$ be a symmetric matrix.  
The operator $\operatorname{vecl}(\mathbf{A})$ denotes the vector obtained by stacking the strictly lower-triangular elements of $\mathbf{A}$ in column-major order, so that \(\operatorname{vecl}(\mathbf{A}) \in \mathbb{R}^{m(m-1)/2}\).  
The operator $\operatorname{diag}$ has two distinct roles:  
(i) When applied to a vector $\mathbf{v} = (v_1, \ldots, v_m)^\top$, it produces a diagonal matrix: $\operatorname{diag}(\mathbf{v}) = \mathrm{diag}(v_1, \ldots, v_m)$.  
(ii) When applied to a square matrix $\mathbf{A}$, it extracts the diagonal elements as a vector: $\operatorname{diag}(\mathbf{A}) = (a_{11}, a_{22}, \ldots, a_{mm})^\top \in \mathbb{R}^m$.
For a symmetric positive definite matrix $\mathbf{A}$ with spectral decomposition $\mathbf{A} = \mathbf{Q}\, \mathrm{diag}(\lambda_1, \ldots, \lambda_m)\, \mathbf{Q}^\top$, where $\mathbf{Q}$ is orthogonal and $\lambda_j > 0$, the matrix exponential and logarithm are defined as
\(e^{\mathbf{A}} = \mathbf{Q}\, \mathrm{diag}(e^{\lambda_1}, \ldots, e^{\lambda_m})\, \mathbf{Q}^\top\) and \(\log \mathbf{A} = \mathbf{Q}\, \mathrm{diag}(\log \lambda_1, \ldots, \log \lambda_m)\, \mathbf{Q}^\top\).

The classical Fisher $z$-transformation,
\(
z = \tfrac{1}{2} \log \frac{1+\rho}{1-\rho},
\)
maps a single correlation $\rho \in (-1,1)$ onto the real line, removing its bounded constraint.  
In practice, however, correlations typically appear as entire matrices that must remain symmetric and positive definite. Extending Fisher’s idea therefore requires a transformation that moves the entire matrix into an unconstrained space, while preserving a valid inverse mapping.

A natural solution is the \emph{matrix logarithm}.  
For a correlation matrix $\bR \in \mathbb{R}^{m \times m}$, the \emph{generalized $z$-transformation} \citep[e.g.,][]{Archakov:2020} is
\begin{equation}\label{eq:gz-transform}
    \bgamma = f(\bR) := \operatorname{vecl}\!\big( \log \bR \big),
\end{equation}
where $\operatorname{vecl}$ stacks the strictly lower-triangular elements of $\log \bR$ into a vector.  
This defines a smooth one-to-one mapping between correlation matrices and an unconstrained Euclidean space, facilitating flexible regression modeling of complex dependence.

Two features make the generalized $z$-transformation particularly appealing for applied work.  
First, it is \emph{order-invariant}: if $\by = \bP \bx$ is a permutation of the variables, then 
\(
\log \bR_y = \bP \, \log \bR_x \, \bP^\top,
\)
so relabeling variables simply permutes the transformed parameters without affecting their values.  
Second, the mapping preserves monotonicity: for any pairwise correlation $\rho_{jk}$ and its transformed counterpart $\gamma_{jk}$,
\(
\frac{\partial \rho_{jk}}{\partial \gamma_{jk}} \ge 0,
\)
ensuring that increases in $\gamma_{jk}$ correspond to increases in $\rho_{jk}$, which aids interpretability \citep[see also][]{hu2024applied}.

These properties ensure both mathematical rigor and practical interpretability, making the generalized $z$-transformation a reliable tool for modeling correlation matrices while preserving positive definiteness.  
By mapping correlations into an unconstrained Euclidean space, it enables regression-style modeling of complex dependence, allowing covariates such as time, location, or group membership to directly inform the correlation structure.  
Building on this, we introduce a flexible parametrization for general correlation matrices.

For each cluster $i$, let the generalized $z$-transformation of its correlation matrix be
\[
\bgamma_i = \operatorname{vecl}\!\big( \log \bR_i \big) 
= \big( \gamma_{ijk} \big), 
\quad i = 1, \dots, n,\; 1 \le k < j \le m_i,
\]  
where $\bR_i = (\rho_{ijk})$ is the within-cluster correlation matrix and $\rho_{ijk} = \operatorname{corr}(y_{ij}, y_{ik})$ denotes the correlation between responses $y_{ij}$ and $y_{ik}$ in the same cluster.

We model each transformed pairwise correlation $\gamma_{ijk}$ as
\begin{equation}\label{eq:corr}
    \gamma_{ijk} = \bw_{ijk}^\top \balpha,
\end{equation}
where $\balpha \in \mathbb{R}^d$ is the \emph{matrix log-correlation parameter} and $\bw_{ijk}$ contains pair-specific covariates capturing factors that explain the correlation between observations $j$ and $k$. 
The covariates $\bw_{ijk}$ are flexible and may include: (i) differences or distances between $\bx_{ij}$ and $\bx_{ik}$; (ii) subgroup indicators (e.g., same higher-level cluster); (iii) other contextual or structural features relevant to the study.
Concrete constructions are illustrated in Section~\ref{sec:real_1} and ~\ref{sec:real_2}.

This formulation offers several advantages. Working with $\bgamma_i = \operatorname{vecl}\!\big(\log \bR_i\big)$ avoids the positive-definiteness constraint on $\bR_i$, maintaining flexibility and interpretability.  
The covariates $\bw_{ijk}$ allow incorporation of rich information---such as spatial/temporal proximity, shared group characteristics, or baseline similarity---directly into the correlation structure.  
Moreover, varying cluster sizes $m_i$ are naturally accommodated, making the model suitable for unbalanced longitudinal or multi-site data.

A key theoretical property supporting this approach is that any symmetric matrix $\bG \in \mathbb{R}^{m \times m}$ can be mapped back to a valid correlation matrix by appropriately adjusting its diagonal entries.  
Specifically, there exists a unique vector $\mathbf{x}^* \in \mathbb{R}^m$ such that
\[
\bR = e^{\,\bG[\mathbf{x}^*]}
\]  
is a correlation matrix, where $\bG[\mathbf{x}^*]$ denotes $\bG$ with $\mathbf{x}^*$ replacing its diagonal.  

The vector $\mathbf{x}^*$ can be computed via a fixed-point iteration \citep{Archakov:2020}:
\[
\mathbf{x}_{(k+1)} \;=\; \mathbf{x}_{(k)} \;-\; \log \!\Big[ \operatorname{diag}\!\Big(e^{\bG[\mathbf{x}_{(k)}]}\Big) \Big],
\]  
initialized from an arbitrary $\mathbf{x}_{(0)}$.  
This provides a constructive algorithm to recover a valid correlation matrix $\bR = f^{-1}(\bgamma)$ from an unconstrained vector $\bgamma$: (i) embed $\bgamma$ into a symmetric matrix $\bG = \operatorname{vecl}^{-1}(\bgamma)$ with arbitrary diagonal entries; (ii) compute the unique diagonal adjustment $\mathbf{x}^*$ using the fixed-point iteration; (iii) obtain the correlation matrix as $\bR = e^{\bG[\mathbf{x}^*]}$.  

We now turn to estimation of the correlation parameter $\balpha_0$.  
Let $\ell_i(\by_i \mid \balpha, \bbeta, \phi)$ denote the density of the multivariate Gaussian distribution $\mathrm{N}(\bmu_i, \bV_i)$, where
\[
\bV_i = \bA_i^{1/2} \bR_i(\balpha) \bA_i^{1/2}.
\]
Given current estimates $(\hat{\bbeta}_n, \hat{\phi}_n)$ from the marginal model, the pseudo-likelihood estimator of $\balpha$ maximizes the Gaussian pseudo-log-likelihood
\begin{equation}\label{eq:p-likelihood}
\sum_{i=1}^n \log \ell_i(\by_i \mid \balpha, \hat{\bbeta}_n, \hat{\phi}_n)
= -\frac{1}{2} \sum_{i=1}^n 
\Big\{
\log |\hat{\bV}_i| + (\by_i - \hat{\bmu}_i)^\top \hat{\bV}_i^{-1} (\by_i - \hat{\bmu}_i)
\Big\} + C,
\end{equation}
where $C$ is a constant independent of $\balpha$ and 
$\hat{\bV}_i = \hat{\bA}_i^{1/2} \bR_i(\balpha) \hat{\bA}_i^{1/2}$.
Defining the standardized residual,
\(
\hat{\bnu}_i = \hat{\bA}_i^{-1/2} (\by_i - \hat{\bmu}_i),
\)
the pseudo-likelihood estimator of $\balpha_0$ can be written as
\[
\hat{\balpha}_n = 
\arg\min_{\balpha} \sum_{i=1}^n 
\Big\{
\log |\bR_i(\balpha)| + \hat{\bnu}_i^\top \bR_i(\balpha)^{-1} \hat{\bnu}_i
\Big\}.
\]

This approach offers a practical, numerically stable way to estimate correlation parameters in the generalized correlation regression framework.  
By working in the unconstrained space of the generalized $z$-transformation, it ensures valid correlation matrices while flexibly incorporating covariates. Under standard regularity conditions, the estimator $\hat{\balpha}_n$ is consistent for $\balpha_0$ and asymptotically normal, providing a sound basis for statistical inference; see the Supplementary Material for technical details. An iterative procedure based on a modified Fisher scoring algorithm can be used to jointly estimate $(\balpha, \bbeta, \phi)$, which is also provided in the Supplementary Material.

\section{Analysis of Modern Prenatal Care Data}\label{sec:real_1}

\subsection{Dataset Overview}

We illustrate the practical utility of the proposed generalized correlation regression (GCR) framework using the motivating modern prenatal care dataset described in Section~\ref{sec:introduction}.  
The dataset comprises \emph{2,449 births} over a five-year period to \emph{1,558 mothers} across \emph{161 communities}, resulting in a naturally nested structure: births within mothers and mothers within communities.  
This example exhibits unbalanced cluster sizes, irregular timing between repeated births, and multiple levels of dependence---features that our methodology is specifically designed to handle.

For clarity of exposition, we performed standard data cleaning and preprocessing: records with missing values were removed, and selected variables were appropriately transformed.  
The final analytic sample includes \emph{2,223 births} from \emph{1,414 mothers} across \emph{160 communities}.  
Table~\ref{table:data_desc} summarizes the transformed and derived variables used in the analysis, along with concise descriptions of their definitions and coding schemes.

\begin{table}[!h]
    \centering
    \resizebox*{1\textwidth}{!}{
        \begin{tabular}{ccc}
            \toprule
            \textbf{Type} & \textbf{Variable}  & \textbf{Description} \\
            \midrule
            \multirow{3}{*}{Cluster Identifier} & kid & Unique identifier for each birth record \\
            & mom & Unique identifier for mother/family unit \\
            & cluster & Unique identifier for community \\
            \midrule
            Response Variable & prenat & Type of prenatal care: $0=\text{traditional}, 1 = \text{modern}$ \\
            \midrule
            \multirow{3}{*}{Individual Covariates} & childAge & Child's age group at survey time: $0=\text{0-2}, 1=\text{3+}$ \\
            & motherAge & Maternal age group at survey time: $0=\text{24-},1=\text{25+}$\\
            & birthOrd & Birth order within the family: $0=\text{1},1=\text{2-3},2=\text{4-6},3=\text{7+}$\\
            \midrule
            \multirow{6}{*}{Family Covariates} & indig & \makecell{Ethnic-linguistic identity:\\ $0 = \text{Ladino}, 1=\text{Indigenous non-Spanish}, 2=\text{Indigenous Spanish}$} \\
            & momEd & Mother's education level: $0=\text{None},1=\text{Primary},2=\text{Secondary+}$\\
            & husEd & Husband's education level: $0=\text{None},1=\text{Primary},2=\text{Secondary+}$ \\
            & husEmpl & \makecell{Husband's employment status:\\ $0=\text{Unskilled},1=\text{Professional},2=\text{Agri-self},3=\text{Agri-empl},4=\text{Skilled}$} \\
            & toilet & Presence of modern toilet in household: $0=\text{None},1=\text{Modern}$ \\
            & TV & TV ownership and viewing frequency: $0=\text{None},1=\text{not daily},2=\text{daily}$ \\
            \midrule
            \multirow{2}{*}{Community Covariates} & pcInd81 & Percentage of indigenous population in community: $[0.006683,0.9958837]$ \\
            & ssDist & Scaled distance to nearest health clinic: $[0,1]$ \\
            \bottomrule
        \end{tabular}
    }
    \caption{Description of Variables of the Prenatal Care Dataset.}
    \label{table:data_desc}
\end{table}

The analysis of this dataset addresses two complementary aims.  
From a scientific perspective, we ask: How is prenatal care utilization---considering timing, frequency, and adequacy---associated with key maternal and neonatal outcomes (e.g., preterm birth or birthweight proxies) after accounting for maternal history and community context? Do these associations vary across mothers or communities, indicating heterogeneity in access or effectiveness?  

From an applied and methodological perspective, we ask: How much correlation arises within mothers and within communities, and how does explicitly modeling this dependence affect effect estimates and their uncertainty compared with conventional working-correlation choices? Which covariates---such as maternal characteristics, inter-birth spacing, or community indicators---help explain the strength of dependence, and does allowing correlation to vary with these features improve model fit and inference?

Let \(y_{ij}\) denote the prenatal care status for birth \(j\) in community \(i\), where \(y_{ij} = 1\) indicates modern care and \(y_{ij} = 0\) indicates traditional care. Let \(p_{ij} = \mathbb{E}(y_{ij}) = \Pr(y_{ij} = 1)\). Following \cite{Pebley:1996}, we specify the mean model as
\begin{align}\label{eq:mean}
\text{logit}(p_{ij}) 
=\, & \beta_0
+ \beta_1\,\text{childAge}_{ij}
+ \beta_2\,\text{motherAge}_{ij}
+ \sum_{k=1}^3 \beta_{2+k}\,\mathbb{I}\big(\text{birthOrd}_{ij} = k\big) \notag\\
&+ \sum_{k=1}^2 \beta_{5+k}\,\mathbb{I}\big(\text{indig}_{ij} = k\big)
+ \sum_{k=1}^2 \beta_{7+k}\,\mathbb{I}\big(\text{momEd}_{ij} = k\big)
+ \sum_{k=1}^2 \beta_{9+k}\,\mathbb{I}\big(\text{husEd}_{ij} = k\big) \notag\\
&+ \sum_{k=1}^4 \beta_{11+k}\,\mathbb{I}\big(\text{husEmpl}_{ij} = k\big)
+ \beta_{16}\,\text{toilet}_{ij}
+ \sum_{k=1}^2 \beta_{16+k}\,\mathbb{I}\big(\text{TV}_{ij} = k\big) \notag\\
&+ \beta_{19}\,\text{pcInd81}_{ij}
+ \beta_{20}\,\text{ssDist}_{ij},
\end{align}
where \(\mathbb{I}(\cdot)\) denotes the indicator function. For each categorical covariate, the baseline category (coded as \(0\)) is omitted to ensure model identifiability.

\subsection{Analysis by Our Correlation Models}

We begin by introducing the standardized residuals:
\begin{equation*}
    \hat{\bvar}_i=(\hat{\varepsilon}_{i1},\dots,\hat{\varepsilon}_{i,m_i})^\top = \hat{\bV}_i^{-1/2}(y_{i1} - \hat{p}_{i1},\dots,y_{i,m_i} - \hat{p}_{i,m_i})^\top,
\end{equation*}
where $\hat{\bV}_i$ is the estimated covariance matrix defined in \eqref{eq:vcov}.
The empirical correlation for a predefined subgroup \(S\) is calculated as:
\begin{equation}\label{eq:e_corr}
    \hat{\rho}_{S} = \frac{1}{N_S} \sum_{(i,j) \neq (t,s)} 
    \hat{\varepsilon}_{ij} \, \hat{\varepsilon}_{ts} \,
    \mathbb{I}\big( (y_{ij}, y_{ts}) \in S \big),
\end{equation}
where \(i,t\) index communities, \(j,s\) index records, and \(N_S\) is the number of observation pairs in subgroup \(S\). We perform \textit{t}-tests on each subgroup to examine the significance of the empirical correlations.

Under a correctly specified model, these empirical correlations are expected to be close to zero. This follows from the fact that $\Var(\hat{\boldsymbol{\varepsilon}}_i) = \hat{\bV}_i^{-1/2} \bV_i \hat{\bV}_i^{-1/2} \approx \bI$ and the assumption of independence between clusters, implying $\mathbb{E}(\hat{\varepsilon}_{ij} \hat{\varepsilon}_{ts}) \approx 0$ for $(i,j) \neq (t,s)$. Hence, any statistically significant deviations from zero indicate a lack of fit. 
Note that when an independent working correlation structure is used, $\hat{\bV}_i$ reduces to a diagonal matrix of marginal variances. In this case, $\hat{\rho}_{S}$ can be directly interpreted as an estimate of the correlation of the outcome $y_{ij}$ within subgroup \(S\). For other working correlation structures, however, $\hat{\rho}_{S}$ primarily serves as a tool for model diagnostics and variable selection.

A primary source of dependence arises from shared community- and family-level influences, expected to induce correlations among prenatal care choices. To visualize these baseline clustering effects, we define four subgroups and calculate the empirical correlation \eqref{eq:e_corr} fitted by the independent GEE.
At the community level, let \(S_1\) denote pairs within the same community (\(i=t\) and \((i,j) \neq (t,s)\)) and \(S_2\) its complement.  
At the family level, \(S_3\) (within-family) and \(S_4\) (between-family) are defined analogously.
The resulting empirical correlations are shown in the upper panel of Table~\ref{table:e_corr_baseline}, all results except $\rho_{S_4}$ are significant.
\begin{table}[!h]
    \centering
    \begin{tabular}{ccccc}
        \toprule
        \textbf{Estimate} & {Within-community} & {Between-community} & {Within-family} & {Between-family} \\
        \midrule
        \textbf{E.Corr.} & $0.1455$ & $-0.0015$ & $0.7764$ & $-0.0007$ \\
        \textbf{$p$-Value} & $0.0000$ & $0.0178$ & $0.0000$ & $0.2934$ \\
        \midrule
        \textbf{E.Corr.} & $0.0000$ & $-0.0004$ & $0.0000$ & $-0.0004$ \\
        \textbf{$p$-Value} & $0.6527$ & $0.5006$ & $0.6351$ & $0.5206$\\
        \bottomrule
    \end{tabular}
    
  \caption{Empirical correlations for community and family clustering. The upper panel shows results for independent GEE; the lower panel shows results for \eqref{eq:corr_init}.}
  \label{table:e_corr_baseline}
\end{table}

Figure~\ref{fig:boxplot} presents boxplots of the empirical correlations for the four subgroups, illustrating the presence of strong within-family correlation and moderate within-community correlation. These results indicate that correlations are stronger within communities and within families, with family-level clustering effects notably more pronounced than community-level effects, which are failed to be captured by the independent GEE.

\begin{figure}[!h]
    \centering
    \includegraphics[width=0.8\textwidth]{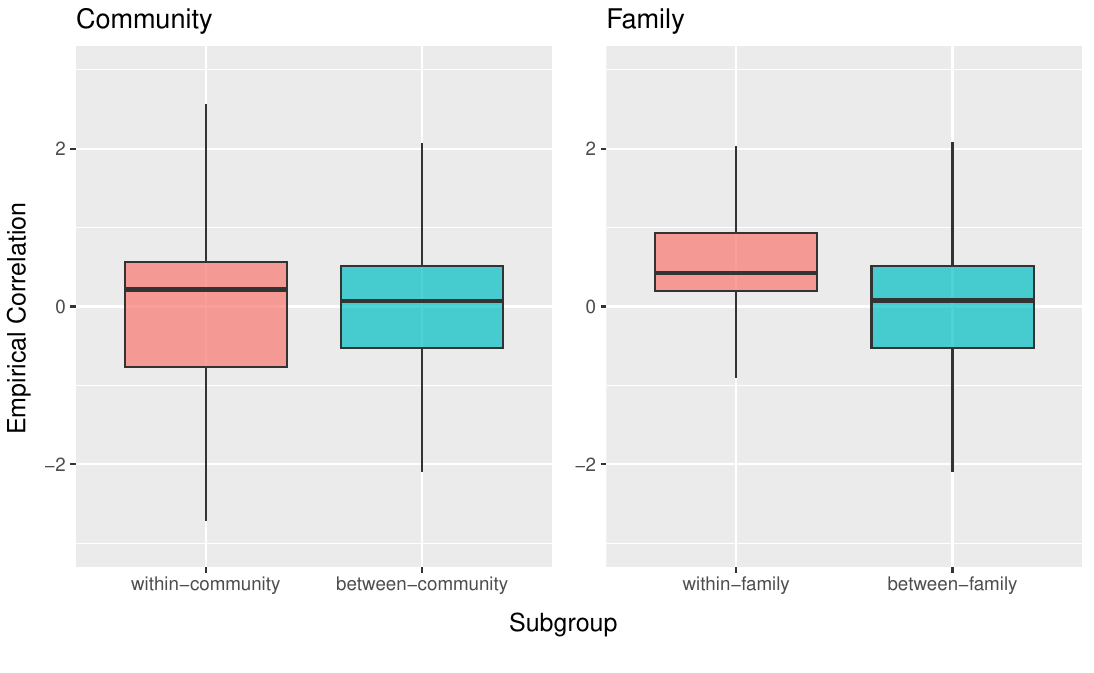}
    \caption{Boxplots of the empirical correlations at the community and family level.}
    \label{fig:boxplot}
\end{figure}  

Motivated by this observation, we fit a simple initial correlation model
\begin{equation}\label{eq:corr_init}
    \gamma_{ijk} = \alpha_0 + \alpha_1 w_{ijk,1},
\end{equation}
where \(w_{ijk,1}=1\) if births $j$ and $k$ within community $i$ belong to the same family, and \(0\) otherwise.
Under this hierarchical structure, the vector \(\bgamma_i\) exhibits a block pattern:  
\[
\bgamma_i = \vecl\!\left( 
\begin{array}{cc:ccc}
* & \alpha_0 + \alpha_1 & \alpha_0 & \alpha_0 & \alpha_0 \\
\alpha_0 + \alpha_1 & * & \alpha_0 & \alpha_0 & \alpha_0 \\
\hdashline
\alpha_0 & \alpha_0 & * & \alpha_0 + \alpha_1 & \alpha_0 + \alpha_1 \\
\alpha_0 & \alpha_0 & \alpha_0 + \alpha_1 & * & \alpha_0 + \alpha_1 \\
\alpha_0 & \alpha_0 & \alpha_0 + \alpha_1 & \alpha_0 + \alpha_1 & * \\
\end{array} 
\right),
\]  
so that \(\alpha_0\) captures the baseline within-community correlation, while \(\alpha_1\) represents the additional within-family correlation. Consequently, the resulting correlation matrix \(\bR_i\) inherits the same block structure. 

The estimated parameters are $\hat{\alpha}_0 = 0.0468$ and $\hat{\alpha}_1 = 0.8617$, both highly significant (p-values $< 0.0001$).  
These results are consistent with the empirical residual correlations and prior findings \citep{Pebley:1996,Rodriguez:2001}, confirming strong family-level and moderate community-level clustering effects in the data.
Again, we examine the empirical correlations \eqref{eq:e_corr} for the four community- and family-level subgroups fitted by our model \eqref{eq:corr_init}, which are shown in the lower panel of Table~\ref{table:e_corr_baseline}. It can be seen that all empirical correlations are now insignificant with values approaching zero, suggesting that our proposed model \eqref{eq:corr_init} successfully captures the two baseline clustering effects.

Beyond the baseline community- and family-level clustering, other features may induce additional correlations. For example, the ethnic-linguistic identifier \textit{indig} links multiple birth records, serving as a natural ethnic-linguistic grouping variable. 

We begin the exploration by computing the empirical correlations \eqref{eq:e_corr} within subgroups defined by family- and individual-level covariates. For ethnic-linguistic identities, three subgroups are defined: a pair of births is considered to belong to the same subgroup only if they share both the same community and the same ethnic-linguistic identity, which corresponds to the three ethnic groups. A similar approach is applied to husband’s employment status. For other family- and individual-level covariates, subgroups are constructed based on absolute differences. For example, in the case of mother’s education level, three subgroups are defined according to whether the absolute difference in education between a birth pair is 0, 1, or 2, while again restricting to pairs within the same community. The empirical correlations are fitted via our initial correlation model \eqref{eq:corr_init}. 
While other covariates tend out to be insignificant, two covariates stand out---ethnic-linguistic identity and husband’s employment status, with results shown in the upper panel of Table \ref{table:e_corr_indig} and \ref{table:e_corr_husempl}. 

\begin{table}[!h]
    \centering
    \begin{tabular}{cccc}
        \toprule
        \textbf{Estimate} & {Ladino} & {Indigenous non-Spanish} & {Indigenous Spanish} \\
        \midrule
        \textbf{E.Corr.} & $0.0040$ & $0.0230$ & $-0.0362$ \\
        \textbf{$p$-Value} & $0.7146$ & $0.3428$ & $0.0817$ \\
        \midrule
        \textbf{E.Corr.} & $0.0029$ & $0.0131$ & $-0.0047$ \\
        \textbf{$p$-Value} & $0.7905$ & $0.5996$ & $0.8120$ \\
        \bottomrule
    \end{tabular}
    
  \caption{Empirical correlations for ethnic-linguistic identity. The upper panel shows results for \eqref{eq:corr_init}; the lower panel shows results for \eqref{eq:corr_family}.}
  \label{table:e_corr_indig}
\end{table}

\begin{table}[!h]
    \centering
    \begin{tabular}{ccccccc}
        \toprule
        \textbf{Estimate} & {Unskilled} & {Professional} & {Agri-self} & {Agri-empl} & {Skilled} \\
        \midrule
        \textbf{E.Corr.} & $0.1002$ & $-0.0168$ & $0.0509$ & $0.0057$ & $0.0066$ \\
        \textbf{$p$-Value} & $0.2947$ & $0.7225$ & $0.0131$ & $0.8292$ & $0.8045$ \\
        \midrule
        \textbf{E.Corr.} & $-0.1372$ & $-0.0409$ & $0.0209$ & $0.0115$ & $-0.0266$ \\
        \textbf{$p$-Value} & $0.4936$ & $0.4115$ & $0.3092$ & $0.6484$ & $0.3255$ \\
        \bottomrule
    \end{tabular}
    
  \caption{Empirical correlations for husband’s employment status. The upper panel shows results for \eqref{eq:corr_init}; the lower panel shows results for \eqref{eq:corr_family}.}
  \label{table:e_corr_husempl}
\end{table}

The empirical correlations are marginally significant (p $< 0.1$) for the indigenous Spanish subgroup and significant (p $< 0.05$) for the agri-self subgroup. This indicates that our initial correlation model \eqref{eq:corr_init} is insufficient to capture the additional correlation patterns beyond the simple two-level hierarchical clustering within these subgroups, thus demanding the need to extend our correlation model by incorporating the two covariates \textit{indig} (ethnic-linguistic identity) and \textit{husEmpl} (husband’s employment status).

To illustrate the flexibility and practical utility of our proposed approach, we incorporate the two covariates accordingly, leading to the following correlation model:
\begin{align}\label{eq:corr_family}
    \gamma_{ijk} = & \ \alpha_0 
    + \alpha_1\, w_{ijk,1} 
    + \sum_{l=0}^2 \alpha_{2+l}\, w_{ijk,2+l} 
    + \sum_{l=0}^4 \alpha_{5+l}\, w_{ijk,5+l},
\end{align}
where the indices \(i,j,k\) follow the same convention as in model~\eqref{eq:corr_init}. The additional correlation covariates are defined as:
\begin{itemize}
    \item \(w_{ijk,2+l} = \mathbb{I}(\text{indig}_{ij} = \text{indig}_{ik} = l)\), for \(l = 0,1,2\), indicating shared ethnic-linguistic identity;
    \item \(w_{ijk,5+l} = \mathbb{I}(\text{husEmpl}_{ij} = \text{husEmpl}_{ik} = l)\), for \(l = 0,1,2,3,4\), indicating shared employment status of the husband.
\end{itemize}
Beyond generic community- and family-level clustering, model~\eqref{eq:corr_family} allows correlation to vary with shared ethnic-linguistic background, socioeconomic similarity---factors relevant to prenatal care access and choices.
The parameter estimates and corresponding $p$-values for model~\eqref{eq:corr_family} are summarized in Table~\ref{table:corr_family}.

\begin{table}[!h]
    \centering
    \sisetup{
        table-format = -1.4,
        table-space-text-post = {***},
        table-align-text-post = false,
        table-number-alignment = center,
    }
    \begin{tabular}{@{} c 
        S[table-format=-1.4, table-column-width=6cm]
        S[table-format=1.4, table-column-width=3cm]
        @{}}
    \toprule
    \textbf{Variable} & \textbf{Estimate} & \textbf{P-value}\\
    \midrule
    Intercept ($\alpha_0$) & 0.0541* & 0.0459 \\
    mom ($\alpha_1$) & 0.8190*** & 0.0000 \\
    \hdashline
    Ladino ($\alpha_2$) & -0.0195 & 0.5512 \\
    Indigenous non-Spanish ($\alpha_3$) & -0.0141 & 0.7405 \\
    Indigenous Spanish ($\alpha_4$) & -0.0626\textbf{.} & 0.0684 \\
    \hdashline
    Unskilled ($\alpha_5$) & 0.3635 & 0.2859 \\
    Professional ($\alpha_6$) & 0.0761 & 0.7812 \\
    Agri-self ($\alpha_7$) & 0.0323\textbf{.} & 0.0881 \\
    Agri-empl ($\alpha_8$) & 0.0122 & 0.6439 \\
    Skilled ($\alpha_9$) & 0.0653 & 0.2667 \\
    \bottomrule
    \end{tabular}

    \medskip
    \flushleft{
        \small 
        Notes: \textbf{.} for $\text{p} < 0.1$, * for $\text{p} < 0.05$, ** for $\text{p} < 0.01$ and *** for $\text{p} < 0.001$.
    }
    
    \caption{Parameter estimates and corresponding p-values for \eqref{eq:corr_family}.}
    \label{table:corr_family}
\end{table}

The estimation results aligns well with the empirical studies above. In model~\eqref{eq:corr_family}, the hierarchical clustering effects \(\hat{\alpha}_0\) and \(\hat{\alpha}_1\) remain significant, while most parameters are insignificant, likely because the strong baseline clustering captures most of their potential variation.
Two exceptions are again \(\hat{\alpha}_4\) (indigenous Spanish ethnic-linguistic group) and \(\hat{\alpha}_7\) (husband agri-self employment), both marginally significant (p $<0.1$), indicating additional dependence beyond generic family- and community-level effects. Moreover, the empirical correlations \eqref{eq:e_corr} for the baseline clustering and all the covariate-based subgroups fitted via our model \eqref{eq:corr_family} are now no longer significant. Specifically, the empirical results for ethnic-linguistic identity and husband’s employment status are shown in the lower panel of Table \ref{table:e_corr_indig} and \ref{table:e_corr_husempl}. This suggests that our model \eqref{eq:corr_family} has already adequately capture all the dependencies associated with family and individual level covariates.

Even after accounting for general family and community clustering, births in the indigenous Spanish population are somewhat less correlated (negative estimated value), whereas mothers in agri-self households show greater similarity in prenatal care choices (positive estimated value).
The finding of weaker within-group correlation among the indigenous Spanish population may be attributed to their higher degree of cultural assimilation, which could promote individual behavioral diversity and reduce uniform group practices.
The stronger correlation observed within agriculturally self-employed households might be explained by occupational isolation. This isolation can foster socioeconomic homogeneity and shared experiences, thereby increasing behavioral consistency in healthcare decisions.
These patterns reflect cultural norms or occupation-related constraints, suggesting potential targets for health interventions or further investigation. 

\subsection{Comparisons with GLMM}

For computational efficiency, we consider a simplified version of model \eqref{eq:corr_family}:
\begin{align}\label{eq:corr_family_2}
    \gamma_{ijk} = & \ \alpha_0 
    + \alpha_1\, w_{ijk,1} 
    + \alpha_{2}\, w_{ijk,2} + \alpha_{3}\, w_{ijk,3} 
    + \alpha_{4}\, w_{ijk,4} + \alpha_{5}\, w_{ijk,5},
\end{align}
where the indices \(i,j,k\) and covariate $w_{ijk,1}$ are the same as in model~\eqref{eq:corr_init}. The additional correlation covariates are defined as:
\begin{itemize}
    \item \(w_{ijk,2} = \mathbb{I}(\text{indig}_{ij} = \text{indig}_{ik} \neq 2)\) and \(w_{ijk,3} = \mathbb{I}(\text{indig}_{ij} = \text{indig}_{ik} = 2)\), indicating shared ethnic-linguistic identity with the indigenous Spanish taken out separately;
    \item \(w_{ijk,4} = \mathbb{I}(\text{husEmpl}_{ij} = \text{husEmpl}_{ik} \neq 2)\) and \(w_{ijk,5} = \mathbb{I}(\text{husEmpl}_{ij} = \text{husEmpl}_{ik} = 2)\), indicating shared employment status of the husband with the agri-self taken out separately.
\end{itemize}
Fitting the simplified model \eqref{eq:corr_family_2} obtains similar results as in model \eqref{eq:corr_family}: $\alpha_0$ and $\alpha_1$ are significantly positive with $\alpha_1$ having the larger absolute value; $\alpha_3$ are $\alpha_5$ are marginally significant with the same signs as in model \eqref{eq:corr_family}, while $\alpha_2$ are $\alpha_4$ are insignificant.


A generalized linear mixed model (GLMM) for this setting can be expressed as
\begin{equation}\label{eq:glmm_covariate}
    \mathrm{logit}(p_{ijklt}) = \text{fixed effects} + u_i + v_{ij} + s_{ik} + c_{il}, 
\end{equation}
where $i$ indexes communities, $j$ indexes families, $k$ indexes ethnic-linguistic identities, $l$ indexes husband’s employment status, and $t$ indexes individual births.  
Here, $p_{ijklt}$ denotes the conditional probability of modern prenatal care;  
$u_i \sim \mathrm{N}(0, \sigma_u^2)$ captures community-level heterogeneity;  
$v_{ij} \sim \mathrm{N}(0, \sigma_v^2)$ captures family-level heterogeneity;
$s_{ik} \sim \mathrm{N}(0, \sigma_s^2)$ models variation of the ethnic-linguistic identity;
$c_{il} \sim \mathrm{N}(0, \sigma_c^2)$ models variation of the husband’s employment status.
This hierarchical formulation explicitly accommodates the two main sources of dependence and the additional covariate-based correlations, providing a natural and interpretable framework aligned with the substantive context of the study.

However, fitting the GLMM \eqref{eq:glmm_covariate} yields statistically insignificant estimates for $\sigma_s$ and $\sigma_c$, which goes against our former empirical findings and analyzing results, suggesting that the model could not discern the variability attributable to ethnic-linguistic identity or husband's employment status.
This apparent discrepancy with our initial empirical findings arises because the GLMM can only estimate a single variance component across all occupations (as a random effect), while our method directly parameterizes the correlation structure itself. This allows us to estimate distinct coefficients (e.g., $\alpha_5$ for agri-self in model \eqref{eq:corr_family_2}) that quantify how each specific occupation category independently modulates the strength of within-cluster dependence.
Consequently, our proposed method enhances interpretability of how specific covariates influence the correlation, and enables fine-grained analysis, which is often infeasible with conventional GEE and GLMM methods.

To substantiate the claim that including \textit{indig} and \textit{husEmpl} is necessary for capturing the correlation structure, we next assess the improvement in model fit. To empirically validate this improvement and rigorously benchmark our approach, we compare its predictive performance against our initial model \eqref{eq:corr_init}, the GLMM and GEE alternatives.
A reduced GLMM is also fitted for comparison with our initial model \eqref{eq:corr_init}:
\begin{equation}\label{eq:glmm_reduced}
\mathrm{logit}(p_{ijt}) = \text{fixed effects} + u_i + v_{ij},
\end{equation}
where the notation follows that of \eqref{eq:glmm_covariate}. Predicted probabilities for the GLMMs are obtained using the \texttt{predict} function. Additionally, we fit a GEE with an exchangeable working correlation structure. Both the GLMM and GEE models incorporate the same set of fixed effects as specified in \eqref{eq:mean}.

We implement a 5-fold cross-validation procedure repeated 15 times, and the
model performance is evaluated using the Brier Score \citep{Brier:1950} and the Log Loss, with \(\text{CV}_1\) and \(\text{CV}_2\) denoting the mean cross-validated scores for the Brier Score and Log Loss, respectively; see Supplementary Material for their definitions.  
For both metrics, smaller values indicate better predictive performance.
To assess relative performance, we perform pairwise $t$-tests comparing the mean CV score differences between our model~\eqref{eq:corr_family_2} and the alternative models.  
The corresponding $t$-statistics and $p$-values are reported in Table~\ref{table:data1_compare}.

\begin{table}[!h]
    \centering
    \begin{tabular}{ccccc}
        \toprule
        \textbf{Criteria} & \textbf{Model \eqref{eq:corr_init}} & \textbf{GEE (Exch)} & \textbf{GLMM \eqref{eq:glmm_covariate}} & \textbf{GLMM \eqref{eq:glmm_reduced}} \\
        \midrule
        $\textbf{CV}_1$ & 0.0001(-5.2233) & 0.0000(-9.2625) & 0.0000(-54.614) & 0.0000(-50.270) \\
        $\textbf{CV}_2$ & 0.0010(-4.1294) & 0.0000(-9.6811) & 0.0000(-52.560) & 0.0000(-48.618) \\
        \bottomrule
    \end{tabular}
    \caption{Model comparison for toenail infection data.}
    \label{table:data1_compare}
\end{table}

The results demonstrate that our proposed model~\eqref{eq:corr_family_2} achieves the lowest mean cross-validation scores with minimal variance across both evaluation metrics, indicating superior predictive performance over all alternative approaches.
In contrast, the GLMM yields substantially higher mean error scores accompanied by greater variability, suggesting not only reduced predictive accuracy but also higher sensitivity to variations in the data.
As summarized in Table~\ref{table:data1_compare}, all pairwise comparisons yield $p$-values below 0.01, providing strong statistical evidence that the performance advantage of our model is highly significant.

Furthermore, while our initial baseline model~\eqref{eq:corr_init} has already outperformed conventional methods, incorporating the covariates \textit{indig} and \textit{husEmpl} in model~\eqref{eq:corr_family_2} leads to a statistically significant improvement in predictive performance, underscoring the value of explicitly modeling covariate-dependent correlation structures.

\section{Additional Real Data Analysis}\label{sec:real_2}

We apply our proposed method to a longitudinal dataset with binary response variables from a randomized clinical trial investigating treatment options for onychomycosis (toenail infection) \citep{DeBacker1998}.  
The trial compared the efficacy of two antifungal therapies---250 mg daily terbinafine versus 200 mg daily itraconazole---under real-world patient adherence patterns.  
A total of 378 patients were randomized to one of the two treatment arms, of which 294 had at least one follow-up record.  

Patients were scheduled for seven follow-up visits at approximately weeks 0, 4, 8, 12, 24, 36, and 48 to monitor clinical improvement.  
At each visit, physicians assessed the infection severity on a binary scale: moderate/severe versus none/mild.  
However, due to variability in adherence and appointment scheduling, the actual visit times (recorded in months) were irregular, and not all patients completed all visits, yielding an unbalanced longitudinal design.   Summary details of the dataset, including the binary response variable, are provided in Table~\ref{table:data_2_desc}.

\begin{table}[!h]
    \centering
    \begin{tabular}{cc}
        \toprule
        \textbf{Variable}  & \textbf{Description} \\
        \midrule
        patientID & Unique identifier for each patient in the trial \\
        \underline{outcome} & Nail infection severity: $0=\text{none or mild}, 1 = \text{moderate or severe}$ \\
        treatment & Type of antifungal treatment: $0=\text{itraconazole},1=\text{terbinafine}$ \\
        time & The time in month when the visit actually took place: $[0,18.5]$ \\
        visit & Number of visit attended: $\{1,2,3,4,5,6,7\}$ \\
        \bottomrule
    \end{tabular}
    \caption{Description of variables of the toenail infection dataset.}
    \label{table:data_2_desc}
\end{table}

From a clinical perspective, this study raises important practical questions:  
(i) how treatment effects evolve over time under irregular visit patterns,  
(ii) whether patient-level or visit-level factors contribute to correlation in treatment responses, and  
(iii) how accounting for such correlation influences the accuracy of treatment effect estimation.  
Our method is particularly well suited for this setting, as it can flexibly model the correlation structure arising from repeated measures with missing and irregularly spaced visits, while providing interpretable inferences on the role of specific covariates.  

Before proceeding to the formal statistical analysis, we conduct an initial exploratory examination of the dataset.  
Figure~\ref{fig:init_explore} summarizes these findings: the left panel shows the distribution of infection status across the two treatment groups, while the right panel depicts the temporal evolution of infection rates over the study period.  
Overall, the itraconazole group consistently exhibits a higher proportion of moderate to severe infections, both in aggregate and at individual visit times.  
In both treatment arms, infection rates decline over time, indicating progressive improvement in patients’ conditions as treatment continues.

\begin{figure}[!h]
    \centering
    \includegraphics[width=0.8\textwidth]{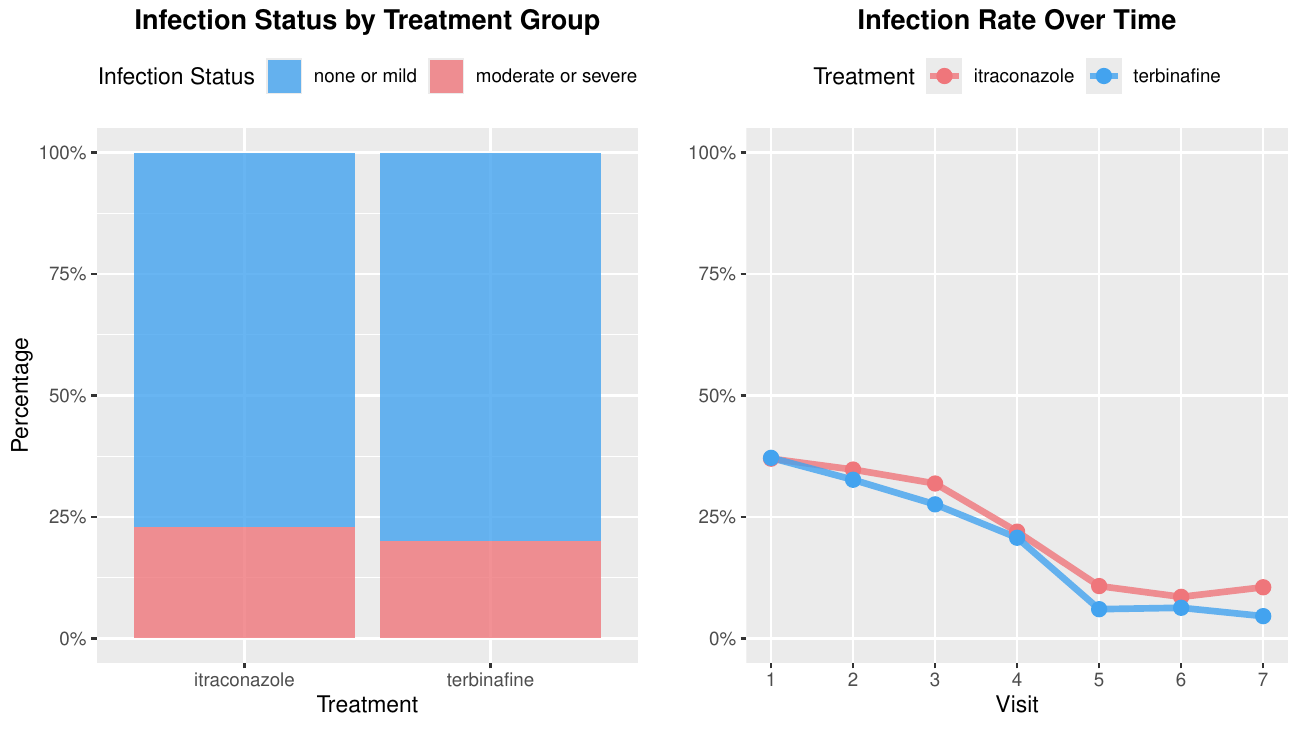}
    \caption{Initial graphical exploration of the toenail infection dataset}
    \label{fig:init_explore}
\end{figure}

Let \(y_{ij}\) denote the binary response indicating the infection severity status for the \(i\)th patient at the \(j\)th visit, and let \(p_{ij} = \mathbb{E}(y_{ij})\).  
We specify the mean model as
\begin{equation}\label{eq:data2_fixed}
    \mathrm{logit}(p_{ij}) = \beta_0 
    + \beta_1\,\text{treatment}_{ij} 
    + \beta_2\,\text{time}_{ij} 
    + \beta_3\,\big(\text{treatment}_{ij} \times \text{time}_{ij}\big),
\end{equation}
where \(\text{treatment}_{ij}\) is an indicator for treatment group, \(\text{time}_{ij}\) is the follow-up time, and the interaction term allows treatment effects to vary over time.

To examine how treatment time influences the correlation structure, we consider two alternative specifications based on different measures of time ordering---one using the discrete \textit{visit} index and the other using the continuous \textit{time} variable. 
Specifically, we fit the following two correlation models:
\begin{align}
    \label{eq:corr_2_visit} \boldsymbol{\gamma}_{ijk} = \alpha_0 + \alpha_1 w_{ijk,1} + \alpha_2 w_{ijk,2},\\
    \label{eq:corr_2_time} \boldsymbol{\gamma}_{ijk} = \tilde{\alpha}_0 + \tilde{\alpha}_1 \tilde{w}_{ijk,1} + \tilde{\alpha}_2 \tilde{w}_{ijk,2},
\end{align}
where \(i\) indexes patients and \(j,k\) index the \(j\)th and \(k\)th visits, respectively. The covariates are defined as:
\begin{itemize}
    \item \(w_{ijk,1} = \tilde{w}_{ijk,1} = \mathbb{I}(\text{treatment}_{i} = 1)\), indicating the baseline correlation difference between the two treatment groups;
    \item \(w_{ijk,2} = \log|\text{visit}_{ij} - \text{visit}_{ik}|\) and \(\tilde{w}_{ijk,2} = \log|\text{time}_{ij} - \text{time}_{ik}|\), representing the logarithm of the visit or time difference, respectively.
\end{itemize}
Parameter estimates for both models, along with their corresponding $p$-values, are reported in Table~\ref{table:data2_corr}.

\begin{table}[h]
    \centering
    \sisetup{
        table-format = -1.4,  
        table-space-text-post = {***}, 
        table-align-text-post = false, 
        table-number-alignment = center,
    }
    \begin{tabular}{@{} c 
        S[table-format=-1.4, table-column-width=2.5cm] 
        S[table-format=1.4, table-column-width=2.0cm]   
        S[table-format=-1.4, table-column-width=2.5cm] 
        S[table-format=1.4, table-column-width=2.0cm]   
        @{}}
    \toprule
    \multirow{2}{*}{\textbf{Variable}} & 
    \multicolumn{2}{c}{\textbf{Model \eqref{eq:corr_2_visit}}} & 
    \multicolumn{2}{c}{\textbf{Model \eqref{eq:corr_2_time}}} \\
    \cmidrule(lr){2-3} \cmidrule(lr){4-5}
    & \textbf{Estimate} & \textbf{P-value} & \textbf{Estimate} & \textbf{P-value} \\
    \midrule
    Intercept ($\beta_0$) & -0.5179** & 0.0013 & -0.4822** & 0.0042\\
    treatment ($\beta_1$) & 0.0313 & 0.8900 & 0.0292 & 0.9014 \\
    time ($\beta_2$) & -0.1530*** & 0.0000 & -0.1703*** & 0.0000\\
    $\text{treatment}\times\text{time}$ ($\beta_3$) & -0.0970* & 0.0361 & -0.0871* & 0.0263\\
    \hdashline
    Intercept ($\alpha_0/\tilde\alpha_0$) & 0.7091*** & 0.0000 & 0.7727*** & 0.0000\\
    treatment ($\alpha_1/\tilde\alpha_1$) & -0.0225 & 0.5678 & 0.0027 & 0.9414 \\
    visit/time ($\alpha_2/\tilde\alpha_2$) & -0.4549*** & 0.0000 & -0.3127*** & 0.0000\\
    \bottomrule
    \end{tabular}

    \medskip
    \flushleft{
        \small
        Notes: \textbf{.} for $\text{p} < 0.1$, * for $\text{p} < 0.05$, ** for $\text{p} < 0.01$ and *** for $\text{p} < 0.001$.
    }
    
    \caption{Parameter estimates and corresponding p-values for \eqref{eq:corr_2_visit} and \eqref{eq:corr_2_time}.}
    \label{table:data2_corr}
\end{table}

The analyses from both correlation models yield similar conclusions in this study.  
In the mean model, although the main treatment effect \(\beta_1\) is not statistically significant, its interaction with time (\(\beta_3\)) is significant, indicating that treatment with \textit{terbinafine} leads to a faster cure rate over time.  
This effect is observed in addition to the overall improvement in infection status over time for both treatments, as reflected by the significant positive effect of \(\beta_2\).

For the correlation models, the results confirm a positive correlation between repeated measurements from the same patient, as captured by \(\alpha_0\) or \(\tilde{\alpha}_0\). 
The small estimates of \(\alpha_1\) and \(\tilde{\alpha}_1\) indicate that receiving the same treatment does not introduce additional correlation beyond the time effect, suggesting that the primary source of correlation is the patient-specific effect---an intuitively reasonable finding in this clinical context. 
This correlation decreases as the time gap between measurements increases, as reflected by the negative estimates of \(\alpha_2\) or \(\tilde{\alpha}_2\).

In this relatively standard setting, conventional methods such as the GEE and GLMM are directly applicable.   
For the GEE, we consider two working correlation structures: exchangeable and AR(1).  
To ensure stable estimation of the AR(1)-GEE model in \texttt{gee}, patients with only a single recorded visit are excluded from the analysis.

For the GEE, we use the same mean model as in~\eqref{eq:data2_fixed}.  
For the GLMM, the fixed effects are identical to those specified on the right-hand side of~\eqref{eq:data2_fixed}, with three independent random effects included:  
\begin{equation*}
    \mathrm{logit}(p_{ijk}) = \text{fixed effects} + u_i + v_j + w_k, 
\end{equation*}
where \(p_{ijk} = \mathbb{E}(y_{ijk})\) is the expected response for the \(i\)th patient, \(j\)th treatment type, and \(k\)th visit.  
Here, \(u_i \sim \mathrm{N}(0, \sigma_u^2)\) captures patient-level heterogeneity, \(v_j \sim \mathrm{N}(0, \sigma_v^2)\) captures treatment-level variation, and \(w_k \sim \mathrm{N}(0, \sigma_w^2)\) models visit-specific effects. Unlike our proposed method, the standard GLMM framework lacks a mechanism to model correlations directly based on continuous variables, unless they are included to have random slopes, which drastically changes the question to be addressed. Thus, to account for temporal correlation, it is restricted to using the discrete variable \textit{visit}.

Fitting the models with both alternative approaches yields quantitatively similar results.  
We then compare their predictive performance against our proposed method.   
To assess predictive accuracy, we employ a 15-repeat stratified 5-fold cross-validation scheme.  
The dataset is first partitioned into two mutually exclusive subsets by treatment type, with random sampling conducted within each subset to preserve the original treatment distribution.
The same cross-validation performance evaluating metrics is used as in Section~\ref{sec:real_1}.
  
Pairwise $t$-tests comparing mean CV score differences between model~\eqref{eq:corr_2_time} and each of the other four models yield the $p$-values and $t$-statistics reported in Table~\ref{table:data2_compare}.  
With the sole exception of the $\text{CV}_1$ comparison against the GEE (Exch) model, all differences are statistically significant negative at the 0.05 level.  
From an applied perspective, these results mean that our approach not only predicts patient outcomes more accurately but also clarifies how the timing of repeated measurements shapes correlation patterns over the treatment period.  
This combination of predictive strength and interpretability provides actionable insights for clinical trial design and patient monitoring strategies, offering a clear advantage over competing methods in both practical utility and statistical performance---likely attributable to its superior adequacy in model fitting.

\begin{table}[!h]
    \centering
    \begin{tabular}{ccccc}
        \toprule
        \textbf{Criteria} & \textbf{Model \eqref{eq:corr_2_visit}} & \textbf{GEE (Exch)} & \textbf{GEE (AR-1)} & \textbf{GLMM} \\
        \midrule
        $\textbf{CV}_1$ & 0.0004(-4.6233) & 0.7984(-0.2603) & 0.0003(-4.7688) & 0.0000(-64.382) \\
        $\textbf{CV}_2$ & 0.0011(-4.1041) & 0.0346(-2.3408) & 0.0000(-6.5205) & 0.0000(-27.065) \\
        \bottomrule
    \end{tabular}
    \caption{Model comparison for toenail infection data.}
    \label{table:data2_compare}
\end{table}

We further illustrate the broad applicability of our framework with additional analysis of the motivating modern prenatal care dataset and analysis of a clustered health count dataset, together with extensive simulation studies. Across these settings, our method consistently outperforms competitors such as the GLMMs in both predictive accuracy and interpretability, demonstrating clear competitive advantages in practical applications. Complete details are provided in the Supplementary Material.

\section{Discussion}\label{sec:conclusion}

This paper introduces a unified modeling framework for analyzing clustered and longitudinal data with complex dependence structures.  
Unlike traditional approaches such as GEE, GLM, or GLMM, our method simultaneously estimates both the marginal mean and the correlation structure in a flexible yet interpretable manner.  
Through extensive simulation studies (details in the Supplementary Material), we demonstrated that the proposed approach consistently outperforms commonly used working correlation structures in terms of both mean and covariance estimation accuracy, particularly when the underlying correlation pattern is heterogeneous or covariate-dependent.  

We further validated the versatility of the framework through three real data analyses: a modern prenatal care dataset exhibiting multilevel clustering, a longitudinal binary dataset from a toenail infection clinical trial, and an additional clustered count dataset.  
Across these applications, our method not only achieved superior predictive performance (e.g., lower Brier Score and Log Loss) but also provided meaningful insights into how family-, community-, and individual-level covariates influence correlation patterns---insights that are often difficult to obtain from standard GEE or GLMM approaches.  

In summary, the proposed framework provides:  
\begin{enumerate}
    \item a principled way to capture multilevel and covariate-dependent correlations,  
    \item straightforward hypothesis testing for correlation parameters, and  
    \item enhanced interpretability of dependence structures alongside improved predictive accuracy.
\end{enumerate}

Future research directions include extending the framework to handle high-dimensional covariates, incorporating nonparametric correlation components for greater flexibility, and developing computationally scalable algorithms for even larger datasets.  
We believe that this approach opens new opportunities for more refined modeling of dependence in clustered and longitudinal data.

\section{Data Availability Statement}

The Modern Prenatal Care Data is openly available in the \texttt{mlmRev} R package at \url{https://CRAN.R-project.org/package=mlmRev}. The Toenail Infection Data is openly available in the \texttt{HSAUR3} R package at \url{https://CRAN.R-project.org/package=HSAUR3}. The Grouse Ticks Data is openly available in the \texttt{lme4} R package at \url{https://CRAN.R-project.org/package=lme4}.

\section{Appendices}

All technical details, as well as the additional real data analyses, are provided in the Supplementary Material.  
Extensive simulation results illustrating the performance of the proposed method under various scenarios are also included therein.

\bibliographystyle{apalike}
\bibliography{paperbib}

\newpage

\begin{center}
{\Large\bf Supplementary Material}
\end{center}

\begin{appendix}

\section{Algorithm}

\subsection{Modified Fisher’s Scoring Algorithm}\label{sec:algorithm}

To jointly estimate $(\balpha, \bbeta, \phi)$, we employ an iterative procedure based on a modified Fisher scoring algorithm that alternates between updates of $\bbeta$, $\balpha$, and $\phi$, combining the score functions from the GEE with the pseudo-likelihood.  
Define the score functions as
\begin{align*}
    \bS_1(\bbeta;\balpha,\phi) &= \sum_{i=1}^n \bD_i^\top \bV_i^{-1} (\by_i - \bmu_i), \\
    \bS_2(\balpha;\bbeta,\phi) &= \sum_{i=1}^n \bW_i^\top 
        \Big(\frac{\partial \brho_i}{\partial \bgamma_i}\Big)^\top 
        \vecl\Big(\bR_i^{-1} \hat{\bR}_i \bR_i^{-1} - \bR_i^{-1}\Big),
\end{align*}
where $\hat{\bR}_i = \bA_i^{-1/2} (\by_i - \bmu_i)(\by_i - \bmu_i)^\top \bA_i^{-1/2}$.  
The corresponding information matrices are
\[
    \bH_1(\balpha,\bbeta,\phi) = \sum_{i=1}^n \bD_i^\top \bV_i^{-1} \bD_i, \quad
    \bH_2(\balpha,\bbeta,\phi) = \sum_{i=1}^n \bW_i^\top 
        \Big(\frac{\partial \brho_i}{\partial \bgamma_i}\Big)^\top 
        \bJ_i 
        \Big(\frac{\partial \brho_i}{\partial \bgamma_i}\Big) \bW_i,
\]
with $\bJ_i = \boldeta_i \boldeta_i^\top$ and $\boldeta_i = \vecl\Big(\bR_i^{-1} \hat{\bR}_i \bR_i^{-1} - \bR_i^{-1}\Big)$.

A key step is evaluating the Jacobian $\frac{\partial \brho}{\partial \bgamma}$, which can be computed efficiently using the formula of \citet{Archakov:2020}:
\[
    \frac{\partial \brho}{\partial \bgamma} 
    = \bE_l \Big( \bI - \bA \bE_d^\top (\bE_d \bA \bE_d^\top)^{-1} \bE_d \Big) \bA (\bE_l + \bE_u)^\top,
\]
where $\bA = \frac{\partial \operatorname{vec}(\bR)}{\partial \operatorname{vec}(\bG)}$, and $\bE_l, \bE_u, \bE_d$ are elimination matrices satisfying
\[
\vecl(\bR) = \bE_l \operatorname{vec}(\bR), \quad
\vecl(\bR^\top) = \bE_u \operatorname{vec}(\bR), \quad
\operatorname{diag}(\bR) = \bE_d \operatorname{vec}(\bR).
\]

Let $\bG = \log \bR$ and consider its spectral decomposition $\bG = \bQ \bLambda \bQ^\top$, where $\bLambda = \operatorname{diag}(\lambda_1, \dots, \lambda_m)$ contains the eigenvalues and $\bQ$ is an orthonormal matrix of eigenvectors. Following \cite{Linton:1995}, the $m^2 \times m^2$ matrix $\bA$ is given by
\[
    \bA = (\bQ \otimes \bQ) \bXi (\bQ \otimes \bQ)^\top,
\]
where $\bXi$ is diagonal with entries
\[
    \bXi_{(j-1) m + k, (j-1) m + k} =
    \begin{cases}
        e^{\lambda_j}, & \text{if } \lambda_j = \lambda_k, \\
        \dfrac{e^{\lambda_j} - e^{\lambda_k}}{\lambda_j - \lambda_k}, & \text{if } \lambda_j \neq \lambda_k,
    \end{cases}
\]
for $j,k = 1,\dots,m$.  
This matrix $\bA$ is symmetric positive definite, ensuring numerical stability in the computation of the Jacobian.

The modified Fisher scoring updates are summarized in Algorithm~\ref{alg:fisher}.  
Because the pseudo-likelihood objective is not globally convex, we incorporate a step size $\lambda \in (0,1]$ to stabilize the iterations.  
To improve convergence in practice, we recommend using multiple initializations, such as marginal GLM estimates for $\bbeta$ and simple structures (e.g., independence or exchangeable) for $\balpha$.

\begin{algorithm}[!h]
    \renewcommand{\algorithmicrequire}{\textbf{Input:}}
    \renewcommand{\algorithmicensure}{\textbf{Output:}}
    \caption{Modified Fisher's Scoring Algorithm}
    \label{alg:fisher}
    \begin{algorithmic}[1]
        \REQUIRE Initial values \((\balpha^{(0)}, \bbeta^{(0)})\), iteration index \(k=1\)
        \ENSURE Estimates \((\balpha, \bbeta, \phi)\)
        \REPEAT
            \STATE Update \(\phi^{(k)}\) using \(\bbeta^{(k-1)}\)
            \STATE Set \(\balpha^{(k,0)} = \balpha^{(k-1)}\), sub-iteration index \(s=1\)
            \REPEAT
                \STATE Update 
                \[
                    \balpha^{(k,s)} = \balpha^{(k,s-1)} + \lambda [\bH_2(\balpha^{(k,s-1)}, \bbeta^{(k-1)}, \phi^{(k)})]^{-1} \bS_2(\balpha^{(k,s-1)}; \bbeta^{(k-1)}, \phi^{(k)})
                \]
                \vspace{-\baselineskip}
                \STATE Increment \(s \leftarrow s + 1\)
            \UNTIL convergence; update \(\balpha^{(k)}\)
            \STATE Update 
            \[
                \bbeta^{(k)} = \bbeta^{(k-1)} + [\bH_1(\balpha^{(k)}, \bbeta^{(k-1)}, \phi^{(k)})]^{-1} \bS_1(\bbeta^{(k-1)}; \balpha^{(k)}, \phi^{(k)})
            \]
            \vspace{-\baselineskip}
            \STATE Increment \(k \leftarrow k + 1\)
        \UNTIL convergence
    \end{algorithmic}
\end{algorithm}

High-dimensional $\boldeta_i$ can slow convergence. To mitigate this, we replace $\bJ_i$ with a pseudo-expectation $\tilde{\E}(\bJ_i)$ capturing the dominant first-order structure while ignoring higher-order dependencies (see Section~\ref{sec:pse-exp}), yielding stable estimates with negligible bias at lower cost.  

Per iteration, the main computational burden is forming and inverting $\bH_1$ and $\bH_2$ ($O(\max(p^3,d^3))$) and performing spectral decompositions of size $m_i$. Parallelization and efficient linear algebra improve scalability. Initializing with simple working correlations and using adaptive step sizes further enhance stability, with convergence monitored via parameter and objective changes.

\subsection{Calculation of the Pseudo-expectation}\label{sec:pse-exp}

We calculate the pseudo-expectation $\tilde{\mathbb{E}}(\bJ_i)$ defined in Section \ref{sec:algorithm}. We need the following lemma: 

\begin{lemma}\label{lem:calculate}
    Suppose $\tilde{\bvar}=(\tilde{\varepsilon}_1,\dots,\tilde{\varepsilon}_d)^\top\in\mathbb{R}^d$ with the correlation matrix $\bR$ and the diagonal marginal variance matrix $\bA$. Let $\bvar=(\varepsilon_1,\dots,\varepsilon_d)^\top=\bA^{-1/2}(\tilde{\bvar}-\E(\tilde\bvar))$.
    Then for any $d\times d$ matrix $\bB$, we have
    \begin{equation*}
        \E(\bvar\bvar^\top\bB\bvar\bvar^\top) = \bR\bB\bR+\bR\bB^\top\bR+\tr(\bB\bR)\bR+\bK(\bB),
    \end{equation*}
    where $[\bK(\bB)]_{ij}=\sum_{k=1}^d\sum_{l=1}^d\bB_{kl}\kappa_{ijkl}$ with $\kappa_{ijkl}=\E(\varepsilon_i\varepsilon_j\varepsilon_k\varepsilon_l)-  \bR_{ij}\bR_{kl}-\bR_{ik}\bR_{jl}-\bR_{il}\bR_{jk}$.
\end{lemma}
\begin{proof}
    For $i=1,\dots,d$, we have
    \begin{equation*}
        \varepsilon_i=\frac{\tilde{\varepsilon}_i-\E(\tilde{\varepsilon}_i)}{\sqrt{\Var(\tilde{\varepsilon}_i)}},
    \end{equation*}
    with $\E(\varepsilon_i)=0,\E(\varepsilon_i^2)=1$.
    Note that
    \begin{equation*}
        \bvar^\top\bB\bvar=\sum\limits_{k=1}^d\sum\limits_{l=1}^d \bB_{kl}\varepsilon_k\varepsilon_l.
    \end{equation*}
    Thus, the $(i,j)$th element of $\E(\bvar\bvar^\top\bB\bvar\bvar^\top)$ is given by
    \begin{equation}\label{eq:rewrite}
        \E\left(\varepsilon_i(\sum\limits_{k=1}^d\sum\limits_{l=1}^d \bB_{kl}\varepsilon_k\varepsilon_l)\varepsilon_j\right)=\sum\limits_{k=1}^d\sum\limits_{l=1}^d\bB_{kl}\E(\varepsilon_i\varepsilon_j\varepsilon_k\varepsilon_l).
    \end{equation}
    The term $\E(\varepsilon_i\varepsilon_j\varepsilon_k\varepsilon_l)$ can be rewritten as 
    \begin{align*}
        \E(\varepsilon_i\varepsilon_j\varepsilon_k\varepsilon_l)= & \  \E(\varepsilon_i\varepsilon_j)\E(\varepsilon_k\varepsilon_l)+\E(\varepsilon_i\varepsilon_k)\E(\varepsilon_j\varepsilon_l)+\E(\varepsilon_i\varepsilon_l)\E(\varepsilon_j\varepsilon_k)+\kappa_{ijkl}\\
        = &\ \bR_{ij}\bR_{kl}+\bR_{ik}\bR_{jl}+\bR_{il}\bR_{jk}+\kappa_{ijkl},
    \end{align*}
    where $\kappa_{ijkl}$ is the fourth-order cumulant.
    Then, we can rewrite the terms of \eqref{eq:rewrite} one by one:
    \begin{align*}
        \text{First term:} & \quad 
        \sum\limits_{k=1}^d\sum\limits_{l=1}^d \bB_{kl}\bR_{ij}\bR_{kl}=\bR_{ij}\sum\limits_{k=1}^d\sum\limits_{l=1}^d \bB_{kl}\bR_{kl}=\bR_{ij}\tr(\bB\bR), \\
        \text{Second term:} & \quad \sum\limits_{k=1}^d\sum\limits_{l=1}^d\bB_{kl}\bR_{ik}\bR_{jl}=\sum\limits_{k=1}^d\sum\limits_{l=1}^d\bR_{ik}\bB_{kl}\bR_{lj}=[\bR\bB\bR]_{ij}, \\
        \text{Third term:} & \quad \sum\limits_{k=1}^d\sum\limits_{l=1}^d\bB_{kl}\bR_{il}\bR_{jk}=\sum\limits_{k=1}^d\sum\limits_{l=1}^d\bR_{il}\bB_{kl}\bR_{kj}=[\bR\bB^\top\bR]_{ij}, \\
        \text{Fourth term:} & \quad 
        \sum\limits_{k=1}^d\sum\limits_{l=1}^d\bB_{kl}\kappa_{ijkl}=[\bK(\bB)]_{ij}.
    \end{align*}
    Thus, we have the matrix form
    \begin{equation*}
        \E(\bvar\bvar^\top\bB\bvar\bvar^\top) = \bR\bB\bR+\bR\bB^\top\bR+\tr(\bB\bR)\bR+\bK(\bB),
    \end{equation*}
    which is the desired result.
\end{proof}

Then, we are prepared for the calculation.

\begin{proof}[Calculation of $\tilde{\mathbb{E}}(\bJ_i)$]
    Let $\boldeta_i=(\eta_{ijk}),\ 1\leqslant k<j\leqslant m_i$.
    Note that
    \begin{equation*}
        \E[(\bR_i^{-1}\hat{\bR}_i\bR_i^{-1})_{jk}]=[\E(\bR_i^{-1}\hat{\bR}_i\bR_i^{-1})]_{jk}=(\bR_i^{-1})_{jk}.
    \end{equation*}
    Let $a_{ijk}$ be the $(j,k)$th element of $\bR_i^{-1}$. Then, the $(\frac{(2m_i-k)(k-1)}{2}+j-k,\frac{(2m_i-s)(s-1)}{2}+l-s)$th element of $\bJ_i$ is given by
    \begin{equation*}
        \E(\eta_{ijk}\eta_{ils})=\E[(\bR_i^{-1}\hat{\bR}_i\bR_i^{-1})_{jk}(\bR_i^{-1}\hat{\bR}_i\bR_i^{-1})_{ls}]-a_{ijk}a_{ils},
    \end{equation*}
    for all $1\leqslant k<j\leqslant m_i$ and $1\leqslant s<l\leqslant m_i$. Let $\bvar_i=\bA_i^{-\frac{1}{2}}(\by_i-\bmu_i)$, and thus $\bR_i^{-1}\hat{\bR}_i\bR_{i}^{-1}=\bR_i^{-1}\bvar_i\bvar_i^\top\bR_i^{-1}$. Let $\bT_{ij}$ be the $j$th column of $\bR_i^{-1}$. We have $(\bR_i^{-1}\hat{\bR}_i\bR_i^{-1})_{jk}=\bT_{ij}^\top\bvar_i\bvar_i^\top\bT_{ik}$. Thus, we have
    \begin{equation*}
        \E[(\bR_i^{-1}\hat{\bR}_i\bR_i^{-1})_{jk}(\bR_i^{-1}\hat{\bR}_i\bR_i^{-1})_{ls}]=\bT_{ij}^\top\E(\bvar_i\bvar_i^\top\bT_{ik}\bT_{il}^\top\bvar_i\bvar_i^\top)\bT_{is}.
    \end{equation*}
    By Lemma~\ref{lem:calculate}, we have
    \begin{equation}\label{eq:rewrite_H2}
        \E(\eta_{ijk}\eta_{ils})=\bT_{ij}^\top[\bR_i\bT_{ik}\bT_{il}^\top\bR_i+\bR_i\bT_{il}\bT_{ik}^\top\bR_i+\tr(\bT_{ik}\bT_{il}^\top\bR_i)\bR_i+\bK(\bT_{ik}\bT_{il}^\top)]\bT_{is}-a_{ijk}a_{ils}.
    \end{equation}
    Since $\bR_i\bR_i^{-1}=\bR_i(\bT_{i1},\bT_{i2},\dots,\bT_{i,m_i})=\bI$, we have $\bR_i\bT_{ik}=\be_k$ and $\bT_{il}^\top\bR_i=(\bR_i\bT_{il})^\top=\be_l^\top$.
    Thus, we can rewrite \eqref{eq:rewrite_H2} as follows:
    \begin{align*}
        \E(\eta_{ijk}\eta_{ils})= & \ \bT_{ij}^\top[\be_k\be_l^\top+\be_l\be_k^\top+\tr(\bT_{ik}\be_l^\top)\bR_i]\bT_{is}+\bT_{ij}^\top[\bK(\bT_{ik}\bT_{il}^\top)]\bT_{is}-a_{ijk}a_{ils}\\
        = & \ a_{ijl}a_{iks}+a_{ijs}a_{ikl} + \bT_{ij}^\top[\bK(\bT_{ik}\bT_{il}^\top)]\bT_{is}.
    \end{align*}
    Then, we need to specify the form of $\bK(\bB)$, where $\bB$ is an arbitrary matrix.
    Regardless of the permutation, we have five different cases of $\E(\varepsilon_{ij}\varepsilon_{ik}\varepsilon_{il}\varepsilon_{is})$ as follows:
    \begin{equation*}
        \E(\varepsilon_{ij}\varepsilon_{ik}\varepsilon_{il}\varepsilon_{is})=
        \begin{cases}
            \E(\varepsilon_{ij}^4), & \text{if}\ j=k=l=s\\
            \E(\varepsilon_{ij}^3\varepsilon_{is}),\quad & \text{if}\ j=k=l\neq s\\
            \E(\varepsilon_{ij}^2\varepsilon_{il}^2),\quad & \text{if}\ j=k\neq l= s\\
            \E(\varepsilon_{ij}^2\varepsilon_{il}\varepsilon_{is}),\quad & \text{if}\ j=k\neq l\neq s\\
            \E(\varepsilon_{ij}\varepsilon_{ik}\varepsilon_{il}\varepsilon_{is}),\quad & \text{if}\ j\neq k\neq l\neq s\\
        \end{cases}.
    \end{equation*}
    Base on the marginal distribution, $\E(\varepsilon_{ij}^4)$ can be calculated directly, which takes the following form:
    \begin{equation*}
        \E(\varepsilon_{ij}^4)=3+\phi \frac{a^{(4)}(\theta_{ij})}{(a''(\theta_{ij}))^2},\ \forall i,j.
    \end{equation*}
    However, since the joint distribution is not specified, the other four expectation terms cannot be obtained. Therefore, we adapt the diagonal cumulant assumption, which ignores higher-order dependencies across components, to approximate the fourth-order cumulant, i.e., $\kappa_{jkls}=0$ if $j,k,l,s$ are not exactly the same. Thus, we have
        \begin{align*}
        \kappa_{jkls}= & \ \E(\varepsilon_{ij}\varepsilon_{ik}\varepsilon_{il}\varepsilon_{is})-[\bR_{i}]_{jk}[\bR_{i}]_{ls}-[\bR_{i}]_{jl}[\bR_{i}]_{ks}-[\bR_{i}]_{js}[\bR_{i}]_{kl}\\
        = & 
        \begin{cases}
            \phi \frac{a^{(4)}(\theta_{ij})}{(a''(\theta_{ij}))^2},\quad & \text{if}\ j=k=l=s\\
            0,\quad & \text{otherwise}\\
        \end{cases}.
    \end{align*}
    Then, $\bK(\bB)$ takes the following form:
    \begin{equation*}
        [\bK(\bB)]_{jk}=\sum\limits_{l=1}^{m_i}\sum_{s=1}^{m_i}\bB_{ls}\kappa_{jkls}=
        \begin{cases}
            \phi\frac{a^{(4)}(\theta_{ij})}{(a''(\theta_{ij}))^2}\bB_{jj},\quad & \text{if}\ j=k\\
            0,\quad & \text{otherwise} \\
        \end{cases},
    \end{equation*}
    and we have 
    \begin{equation*}
        \bT_{ij}^\top[\bK(\bT_{ik}\bT_{il}^\top)]\bT_{is}=\phi\sum\limits_{t=1}^{m_i}\frac{a^{(4)}(\theta_{it})}{(a''(\theta_{it}))^2}a_{itj}a_{its}a_{itk}a_{itl}.
    \end{equation*}
    Thus, the expression of \eqref{eq:rewrite_H2} under the diagonal cumulant assumption is shown as follows:
    \begin{equation}\label{eq:pseudo_exptation}
        \E(\eta_{ijk}\eta_{ils})=a_{ijl}a_{iks}+a_{ijs}a_{ikl} +\phi\sum\limits_{t=1}^{m_i}\frac{a^{(4)}(\theta_{it})}{(a''(\theta_{it}))^2}a_{itj}a_{its}a_{itk}a_{itl}.
    \end{equation}
    However, in practice, the above result usually gives a much larger step length resulting in the instability. To overcome this, we simply replace the $a_{itj}$ in the last term of \eqref{eq:pseudo_exptation} by $b_{itj}$ for adjustment, where $b_{itj}$ denotes the $(t,j)$th element of $\bR_i^{-1/2}$.
    Therefore, the final expression of the pseudo-expectation takes the following form:
    \begin{equation*}
        \tilde{\E}(\eta_{ijk}\eta_{ils})=a_{ijl}a_{iks}+a_{ijs}a_{ikl} +\phi\sum\limits_{t=1}^{m_i}\frac{a^{(4)}(\theta_{it})}{(a''(\theta_{it}))^2}b_{itj}b_{its}b_{itk}b_{itl}.
    \end{equation*}
    Then, it is easy to obtain $\tilde{\E}(\bJ_i)$ in the matrix form.
\end{proof}

\section{Theory}\label{sec:theory}

We now establish the asymptotic properties of the proposed estimators. To this end, we impose a set of regularity conditions that ensure identifiability, stability of the estimating equations, and the validity of asymptotic arguments. 

\subsection{Asymptotic Property} 

\medskip
\noindent{\textbf{Regularity conditions.}}
We assume:  
\begin{itemize}
    \item[(A1)] The working marginal model and the correlation model are correctly specified.  
    \item[(A2)] The dimensions \(p\) and \(d\) of the covariates \(\bx_{ij}\) and \(\bw_{ij}\), respectively, are fixed as \(n \to \infty\); the covariates and the maximum cluster size \(\max_{1 \le i \le n} m_i\) are uniformly bounded.  
    \item[(A3)] The parameter space \(\bOmega\) of \((\balpha^\top, \bbeta^\top, \phi)^\top\) is a compact subset of \(\mathbb{R}^{d + p + 1}\), and the true parameter value \((\balpha_0^\top, \bbeta_0^\top, \phi_0)^\top\) lies in the interior of \(\bOmega\).  
\end{itemize}

Condition (A1) ensures the model is well-specified, while (A2) restricts the growth of covariate and cluster dimensions to control the complexity of the problem. The compactness assumption (A3) prevents pathological cases where parameters may diverge or approach the boundary of the parameter space.  

\medskip
\noindent{\textbf{Asymptotic behavior of \(\hat{\bbeta}_n\).}}
We first consider the regression coefficients \(\bbeta\), which are estimated through the GEE component of our procedure. The following result, adapted from \citet{Xie:2003}, shows that under mild assumptions on the correlation and scale parameters, the estimator \(\hat{\bbeta}_n\) inherits the usual large-sample properties of GEE estimators.  

\begin{thm}[Theorem 4, \cite{Xie:2003}]\label{thm:beta_con}
    Suppose that \(\hat{\phi}_n - \phi_0 = O_p(n^{-1/2})\) when \(\bbeta_0\) is known and \(\hat{\balpha}_n - \balpha_0 = O_p(n^{-1/2})\) when \(\bbeta_0\) and \(\phi_0\) are known. Then, under regularity conditions (A1)--(A3) and the conditions in \cite{Xie:2003}, the estimator \(\hat{\bbeta}_n\) is weakly consistent for \(\bbeta_0\) and satisfies the asymptotic normality
    \[
        \bM_{1,n}^{1/2}(\hat{\bbeta}_n - \bbeta_0) \overset{d}{\rightarrow} \mathrm{N}(0, \bI),
    \]
    where \(\bM_{1,n} = \bH_1(\balpha_0,\bbeta_0,\phi_0)\) is defined in Section~\ref{sec:algorithm}.
\end{thm}

Intuitively, as long as the correlation and scale parameters are consistently estimated at a sufficiently fast rate, the regression coefficient estimates behave asymptotically as if those nuisance parameters were known.  

\medskip
\noindent{\textbf{Identifiability for \(\balpha\).}}
We now turn to the correlation parameter \(\balpha\). A key ingredient for the consistency of \(\hat{\balpha}_n\) is the identifiability of the pseudo-likelihood function. For Gaussian working models, this follows from a well-known property of the Kullback-Leibler divergence:  
\[
    \E_0 \left( \log \frac{\ell_i(\by_i \mid \balpha, \bbeta_0, \phi_0)}{\ell_i(\by_i \mid \balpha_0, \bbeta_0, \phi_0)} \right) \le 0,
\]
with equality if and only if \(\balpha = \balpha_0\) \citep{Gourieroux:1984}.  

To establish consistency and asymptotic normality for \(\hat{\balpha}_n\), we additionally require the following conditions, which mirror classical results for maximum likelihood estimation of correlation parameters:  

\begin{itemize}
    \item[(B1)] For any \(\balpha \neq \balpha_0\), the average log-likelihood ratio converges to a negative number:
    \[
        \frac{1}{n} \sum_{i=1}^n \E_0 \left( \log \frac{\ell_i(\by_i \mid \balpha, \bbeta_0, \phi_0)}{\ell_i(\by_i \mid \balpha_0, \bbeta_0, \phi_0)} \right) \;\to\; \text{a strictly negative limit.}
    \]
    \item[(B2)] The Fisher information for \(\balpha\) is well-behaved, i.e.,
    \[
        \frac{1}{n} \sum_{i=1}^n \E \!\left( 
            \frac{\partial \log \ell_i(\by_i \mid \balpha_0, \bbeta_0, \phi_0)}{\partial \balpha} 
            \frac{\partial \log \ell_i(\by_i \mid \balpha_0, \bbeta_0, \phi_0)}{\partial \balpha^\top} 
        \right)
        \;\to\; \text{a positive definite matrix.}
    \]
    \item[(B3)] The expected Hessian matrix of \(\log\ell_i\) for \(\balpha\) is well-behaved, i.e.,
    \[
        \frac{1}{n} \sum_{i=1}^n \E \!\left( 
            \frac{\partial^2 \log \ell_i(\by_i \mid \balpha_0, \bbeta_0, \phi_0)}{\partial \balpha\partial \balpha^\top} 
        \right)
        \;\to\; \text{a non-singular matrix.}
    \]
\end{itemize}

Condition (B1) ensures that the true parameter \(\balpha_0\) uniquely maximizes the expected log-pseudo-likelihood, while (B2) and (B3) guarantee sufficient curvature for asymptotic normality.  

\medskip
\noindent{\textbf{Asymptotic behavior of \(\hat{\balpha}_n\).}}
With these conditions in place, we can show that the pseudo-likelihood estimator of \(\balpha\) is both consistent and asymptotically efficient in the usual sense:  

\begin{thm}\label{thm:alpha_con}
    Under regularity assumptions (A1)--(A3) and identifiability conditions (B1)--(B3), and given \(\bbeta_0\) and \(\phi_0\), the estimator \(\hat{\balpha}_n\) is strongly consistent for \(\balpha_0\) and satisfies
    \[
        \bM_{2,n}(\hat{\balpha}_n - \balpha_0) \;\overset{d}{\rightarrow}\; \mathrm{N}\!\big(0, \bI\big),
    \]
    where \(\bM_{2,n}=\bH_2(\balpha_0,\bbeta_0,\phi_0)^{-\frac{1}{2}}\sum_{i=1}^n\frac{\partial^2\log \ell_i(\by_i|\balpha_0,\bbeta_0, \phi_0)}{\partial\balpha\partial\balpha^\top}\).
\end{thm}
The exact expression of the leading matrix component of \(\bM_{2,n}\) matches \(\bH_2\) defined in Section~\ref{sec:algorithm}. However, due to the high complexity of the Hessian matrix, we employ the numerical method to obtain its approximation. Denote
\begin{equation*}
     \ba_i(h)=\frac{\bS_2(\balpha+h\be_i;\bbeta,\phi)-\bS_2(\balpha-h\be_i;\bbeta,\phi)}{2h},
\end{equation*}
where $\bS_2$ is the score function defined in Section~\ref{sec:algorithm} and $h\in\mathbb{R}^+$ denotes a sufficiently small constant.
Let \(\bA(h)=(\ba_1(h),\ba_2(h),\dots,\ba_d(h))\).
Lemma 4.2.2 in \cite{dennis:1996} shows that, under mild conditions, 
\begin{equation*}
    \left\|\hat{\bA}(h)-\sum_{i=1}^n\frac{\partial^2\log \ell_i(\by_i|\balpha,\bbeta, \phi)}{\partial\balpha\partial\balpha^\top}\right\|=O(nh),
\end{equation*}
where \(\hat{\bA}(h)=(\bA(h)+\bA(h)^\top)/2\) since the approximation should be symmetric. Thus, the Hessian matrix can be well approximated by \(\hat{\bA}(h)\) for \(h\) sufficiently small.

\medskip
\noindent{\textbf{Discussion.}}
Theorems~\ref{thm:beta_con} and~\ref{thm:alpha_con} together show that the proposed two-stage procedure enjoys desirable large-sample guarantees. In particular, the estimation of \(\bbeta\) is asymptotically insensitive to small misspecification of \(\balpha\), while the pseudo-likelihood estimator of \(\balpha\) remains consistent due to the identifiability of the Gaussian working model. Importantly, these results hold under bounded cluster sizes and fixed covariate dimensions; extensions to high-dimensional or diverging-cluster settings would require additional technical tools, which we leave for future work. 

\subsection{Proof of Theorem~\ref{thm:alpha_con}}

\begin{proof}[Proof of Theorem \ref{thm:alpha_con}]
    Let $\bpi_0=(\bbeta^\top_0,\phi_0)^\top$ and define
    \begin{equation*}
         \psi_i(\by_i|\balpha,\bpi_0)=\log \ell_i(\by_i|\balpha,\bpi_0) - \log \ell_i(\by_i|\balpha_0,\bpi_0) = \log \frac{\ell_i(\by_i|\balpha,\bpi_0)}{\ell_i(\by_i|\balpha_0,\bpi_0)}.
    \end{equation*}
    Because of the compactness of the parameter space $\bOmega$ and the boundedness of the covariates, 
    it is easy to show that there exists a constant $K$ independent of $\balpha$ such that $\Var_0[\psi_i(\by_i|\balpha,\bpi_0)]<K$, for all $i$.
    Let 
    \begin{align*}
        \mathscr{L}_n(\balpha,\bpi_0)& =\sum_{i=1}^n\log \ell_i(\by_i|\balpha,\bpi_0),\\
        \bar{\mathscr{L}}_n(\balpha,\bpi_0) & =\frac{1}{n}\mathscr{L}_n(\balpha,\bpi_0).
    \end{align*}
    By the strong law of large numbers, we have
    \begin{equation*}
        \bar{\mathscr{L}_n}(\balpha,\bpi_0)-\bar{\mathscr{L}_n}(\balpha_0,\bpi_0)-\frac{1}{n}\sum\limits_{i=1}^n\E_0\left(\log\frac{\ell_i(\by_i|\balpha,\bpi_0)}{\ell_i(\by_i|\balpha_0,\bpi_0)}\right)\overset{a.s.}\rightarrow0.
    \end{equation*}
    Because the constant $K$ is independent of $\balpha$ and by our assumptions, it can be shown that the convergence above is uniform in $\balpha$. It can also be shown that $\frac{1}{n}\sum_{i=1}^n\E_0\left(\log\frac{\ell_i(\by_i|\balpha,\bpi_0)}{\ell_i(\by_i|\balpha_0,\bpi_0)}\right)$ is equicontinuous in $\balpha$ and converges to a finite number $K(\balpha)$ uniformly.
    Thus, we have
    \begin{equation*}
        \bar{\mathscr{L}_n}(\balpha,\bpi_0)-\bar{\mathscr{L}_n}(\balpha_0,\bpi_0)\overset{a.s.}\rightarrow K(\balpha)
    \end{equation*}
    uniformly in $\balpha$. By condition (B1), we have $K(\balpha)<0$ for $\balpha\neq\balpha_0$. 

    For a contradiction, assume that there exists a set of positive probability where $\{\hat{\balpha}_n\}_{n=1}^\infty$ does not converge to $\balpha_0$. Because of the compactness of the parameter space $\bOmega$, there exists a subsequence $\{\hat{\balpha}_m\}_{m=1}^\infty$ and a limit point $\tilde{\balpha}\neq\balpha_0$ such that $\hat{\balpha}_m\rightarrow\tilde{\balpha}$. Because $\hat{\balpha}_m$ produces a maximum for every $m$, we have
    \begin{equation*}
        \bar{\mathscr{L}}_m(\hat{\balpha}_m,\bpi_0)-\bar{\mathscr{L}}_m(\balpha_0,\bpi_0)\geqslant0.
    \end{equation*}
    It is easy to see that $K(\balpha)$ is in continuous in $\balpha$. Then by uniform convergence and continuity of the limit, we have $K(\tilde{\balpha})\geqslant0$,
    which is a contradiction!
    Thus, we conclude that $\hat{\balpha}_n$ is strongly consistent for $\balpha_0$.

    Next, we give the proof of the asymptotic normality. Because $\hat{\balpha}_n$ is the maximum of $\mathscr{L}_n(\balpha,\bpi_0)$, we have
    \begin{align*}
        0 & =\sum\limits_{i=1}^n\frac{\partial\log \ell_i(\by_i|\hat{\balpha}_n,\bpi_0)}{\partial \balpha}\\
        & =\sum\limits_{i=1}^n\frac{\partial\log \ell_i(\by_i|\balpha_0,\bpi_0)}{\partial\balpha}+\sum\limits_{i=1}^n\frac{\partial^2\log \ell_i(\by_i|\tilde{\balpha}_n,\bpi_0)}{\partial\balpha\partial\balpha^\top}(\hat{\balpha}_n-\balpha_0),
    \end{align*}
    where $\tilde{\balpha}_n$ lies between $\hat{\balpha}_n$ and $\balpha_0$. Thus, we have
    \begin{equation*}
        \sqrt{n}(\hat{\balpha}_n-\balpha_0)=\left(-\frac{1}{n}\sum\limits_{i=1}^n\frac{\partial^2\log \ell_i(\by_i|\tilde{\balpha}_n,\bpi_0)}{\partial\balpha\partial\balpha^\top}\right)^{-1}\frac{1}{\sqrt{n}}\sum\limits_{i=1}^n\frac{\partial\log \ell_i(\by_i|\balpha_0,\bpi_0)}{\partial\balpha}.
    \end{equation*}
    Denote
    \begin{align*}
        \tilde{\boldsymbol{\mathcal{H}}}_n(\balpha) & =-\frac{1}{n}\sum\limits_{i=1}^n\frac{\partial^2\log \ell_i(\by_i|\balpha,\bpi_0)}{\partial\balpha\partial\balpha^\top},\\
        \boldsymbol{\mathcal{H}}_n(\balpha) & =-\frac{1}{n}\sum\limits_{i=1}^n\E\left(\frac{\partial^2\log \ell_i(\by_i|\balpha,\bpi_0)}{\partial\balpha\partial\balpha^\top}\right),\\
        \tilde{\boldsymbol{\mathcal{I}}}_n(\balpha) & =\frac{1}{n}\sum\limits_{i=1}^n\frac{\partial\log \ell_i(\by_i|\balpha,\bpi_0)}{\partial\balpha}\frac{\partial\log \ell_i(\by_i|\balpha,\bpi_0)}{\partial\balpha^\top},\\
        \boldsymbol{\mathcal{I}}_n(\balpha) & =\frac{1}{n}\sum\limits_{i=1}^n\E\left(\frac{\partial\log \ell_i(\by_i|\balpha,\bpi_0)}{\partial\balpha}\frac{\partial\log \ell_i(\by_i|\balpha,\bpi_0)}{\partial\balpha^\top}\right).
    \end{align*}
    Because of the compactness of the parameter space $\Omega$ and the boundedness of the covariates, we have $\tilde{\boldsymbol{\mathcal{H}}}_n(\balpha)-\boldsymbol{\mathcal{H}}_n(\balpha)\overset{a.s.}{\rightarrow}0$ and $\tilde{\boldsymbol{\mathcal{I}}}_n(\balpha)-\boldsymbol{\mathcal{I}}_n(\balpha)\overset{a.s.}{\rightarrow}0$ both uniformly in $\balpha$. 
    It can also be shown that $\boldsymbol{\mathcal{H}}_n(\balpha)$ is equicontinuous in $\balpha$. 
    By our assumptions, we have $\boldsymbol{\mathcal{H}}_n(\balpha_0)\overset{a.s.}{\rightarrow}\boldsymbol{\mathcal{H}}(\balpha_0)$, and thus we have $\boldsymbol{\mathcal{H}}_n(\tilde\balpha_n)\overset{a.s.}{\rightarrow}\boldsymbol{\mathcal{H}}(\balpha_0)$ since $\hat\balpha_n\overset{a.s.}{\rightarrow}\balpha_0$ and $\tilde\balpha_n$ lies between them. Combining the above results, we have
    $\tilde{\boldsymbol{\mathcal{H}}}_n(\balpha_0)\tilde{\boldsymbol{\mathcal{H}}}_n(\tilde{\balpha}_n)^{-1}\overset{a.s.}{\rightarrow}\bI$.

    By our assumptions and the central limit theorem, we have
    \begin{equation*}
        \frac{1}{\sqrt{n}}\sum\limits_{i=1}^n\frac{\partial\log \ell_i(\by_i|\balpha_0,\bpi_0)}{\partial\balpha}\overset{d}{\rightarrow}\mathrm{N}(0,\boldsymbol{\mathcal{I}}(\balpha_0)),
    \end{equation*}
    where $\boldsymbol{\mathcal{I}}(\balpha_0)=\lim\limits_{n\rightarrow\infty}\boldsymbol{\mathcal{I}}_n(\balpha_0)$
    is a positive definite matrix.
    Denote
    \begin{equation*}
        \bM_{2,n}=-\sqrt{n}\tilde{\boldsymbol{\mathcal{I}}}_n^{-\frac{1}{2}}(\balpha_0)\tilde{\boldsymbol{\mathcal{H}}}_n(\balpha_0),
    \end{equation*}
    and thus we have
    \begin{equation*}
        \bM_{2,n}(\hat{\balpha}_n-\balpha_0)\overset{d}{\rightarrow}\mathrm{N}(0,\bI)
    \end{equation*}
    by Slutsky's theorem and the symmetry of the Gaussian distribution.
\end{proof}

\section{Simulations}\label{sec:sim}

In this section, we investigate the finite-sample performance of the proposed method through simulation studies. We focus on evaluating the accuracy of parameter estimation, the empirical coverage of confidence intervals, and the impact of sample size and distributional settings.

\subsection{Performance of the Proposed Approach}  \label{sec:study1}
We first validate the proposed method by generating data from three canonical exponential family distributions: Gaussian, Poisson, and Bernoulli. For each distribution, the corresponding canonical link function is used.  

Two global settings are considered for Study~1:  
(i) the cluster size \(m_i\) for each group is generated as \(m_i \sim \operatorname{Bin}(6, 0.8) + 1\), ensuring varying but bounded cluster sizes;  
(ii) the total sample size is varied over \(n = 100, 200,\) and \(400\), with each configuration repeated 1000 times to obtain stable Monte Carlo estimates.  

For each simulated dataset, we compute:  
- the mean absolute errors (MAE) for all parameters, along with their Monte Carlo standard errors in parentheses;  
- the empirical coverage of the nominal \(95\%\) confidence intervals for \(\balpha\) and \(\bbeta\).  
For ease of presentation, all results are multiplied by 100.  

\medskip
\noindent{\textbf{Gaussian data generation.}}
For each cluster \(i\), we generate the covariates \(\bx_{ij} = (1, x_{ij,1}, x_{ij,2})^\top\), where \((x_{ij,1}, x_{ij,2})^\top\) follows a bivariate Gaussian distribution with mean \(\boldsymbol{0}\), marginal variance 1, and correlation 0.5.  
The correlation covariates are constructed as  
\[
\bw_{ijk} = \big(1, u_{ij} - u_{ik}, (u_{ij} - u_{ik})^2 \big)^\top,
\]  
where \(u_{ij} \sim \operatorname{U}(0,1)\) are independent draws.  
We specify the linear predictors as  
\begin{align*}
    \theta_{ij} &= \beta_0 + x_{ij,1} \beta_1 + x_{ij,2} \beta_2, 
    && \forall\; 1 \leqslant i \leqslant n,\; 1 \leqslant j \leqslant m_i,\\
    \gamma_{ijk} &= \alpha_0 + w_{ijk,1} \alpha_1 + w_{ijk,2} \alpha_2, 
    && \forall\; 1 \leqslant i \leqslant n,\; 1 \leqslant k < j \leqslant m_i.
\end{align*}

Since \(a'(x) = x\) and \(h(x) = x\) for the Gaussian family, the canonical link function is the identity link, \(\mu(x) = x\). Thus, given \(\{\theta_{ij}\}\) and \(\{\gamma_{ijk}\}\), the cluster-level response \(\by_i\) follows a multivariate Gaussian distribution with the corresponding mean and covariance matrix determined by the model.  
The true parameter values are fixed as  
\[
\bbeta_0 = (1.0, -0.5, 0.5)^\top, 
\quad \balpha_0 = (0.2, -0.2, 0.3)^\top, 
\quad \phi_0 = 1.
\] 

Table~\ref{table:study_1_Gaussian} reports the results for the Gaussian case. The upper panel shows MAE and standard errors, while the lower panel reports empirical coverage. 
\begin{table}[!htbp]
    \centering
    \resizebox{\textwidth}{!}{
        \begin{tabular}{cccccccc}
            \toprule
            \multicolumn{8}{c}{\textbf{Gaussian Distribution}}\\
            Size & $\alpha_0$ & $\alpha_1$ & $\alpha_2$ & $\beta_0$ & $\beta_1$ & $\beta_2$ & $\phi$\\
            \midrule
            $n=100$ & 2.64(3.37) & 4.61(5.86) & 9.41(11.90) & 5.45(6.89) & 3.16(3.99) & 3.02(3.78) & 6.43(8.05)\\
            $n=200$ & 1.92(2.38) & 3.36(4.28) & 6.69(8.38) & 3.87(4.81) & 2.23(2.78) & 2.19(2.75) & 4.42(5.49)\\ 
            $n=400$ & 1.27(1.60) & 2.32(2.90) & 4.73(5.97) & 2.73(3.42) & 1.57(1.95) & 1.56(1.92) & 3.14(3.88)\\
            \midrule
            $n=100$ & 90.9\% & 93.4\% & 94.9\% & 94.5\% & 94.7\% & 95.7\% & --\\
            $n=200$ & 91.8\% & 92.7\% & 93.4\% & 94.6\% & 95.6\% & 94.9\% & -- \\ 
            $n=400$ & 92.3\% & 93.5\% & 94.8\% & 94.2\% & 95.8\% & 95.9\% & -- \\
            \bottomrule
        \end{tabular}
    }
    \caption{Simulation results for the Gaussian distribution in Study~1. The upper panel shows mean absolute errors (standard errors in parentheses); the lower panel reports empirical coverage of nominal \(95\%\) confidence intervals. All values are multiplied by 100.}
    \label{table:study_1_Gaussian}
\end{table}
As expected, the MAE decreases steadily as \(n\) increases, demonstrating consistency of the estimators. The coverage probabilities for \(\bbeta\) remain close to the nominal level, while those for \(\balpha\) show slight under-coverage for small \(n\), which improves with larger sample sizes. The scale parameter \(\phi\) is estimated reasonably well, with decreasing variability as \(n\) grows.  

\medskip
\noindent{\textbf{Poisson data generation.}}
For the Poisson distribution, we adopt the data-generating algorithm proposed by \cite{Yahav:2012}, which is specifically designed to handle correlated count data. The covariates \(\bx_{ij}\) and \(\bw_{ijk}\), as well as the linear predictors \(\theta_{ij}\) and \(\gamma_{ijk}\), are specified in the same way as in the Gaussian case.  

For the Poisson family, we have \(a'(x) = e^x\) and \(h(x) = x\), which implies that the canonical mean function is the exponential link, \(\mu(x) = e^x\). The true parameters are set to  
\[
\bbeta_0 = (1.0, -0.5, 0.5)^\top, \quad 
\balpha_0 = (0.2, -0.2, 0.3)^\top.
\]  
Unlike the Gaussian case, the Poisson variance is fully determined by the mean, so the nuisance scale parameter \(\phi\) is fixed at \(\phi = 1\) and does not require estimation.  

A technical point is that the Yahav-Shmueli algorithm does not guarantee feasibility for arbitrary correlation matrices. In practice, if the generated correlation matrix fails to satisfy the algorithm's constraints (e.g., non-positive definiteness or incompatibility with Poisson marginal moments), we simply repeat the generation until a feasible sample is obtained. This ensures all simulated datasets are valid.  

The Monte Carlo performance metrics for the Poisson case are summarized in Table~\ref{table:study_1_Poisson}. As before, the upper panel reports mean absolute errors (MAE) with Monte Carlo standard errors in parentheses, while the lower panel shows empirical coverage probabilities of the nominal 95\% confidence intervals.  

\begin{table}[!htbp]
    \centering
    \begin{tabular}{ccccccc}
        \toprule
        \multicolumn{7}{c}{\textbf{Poisson Distribution}}\\
        Size & $\alpha_0$ & $\alpha_1$ & $\alpha_2$ & $\beta_0$ & $\beta_1$ & $\beta_2$\\
        \midrule
        $n=100$ & 2.48(3.08) & 4.93(6.17) & 10.64(13.10) & 3.48(4.40) & 1.95(2.44) & 1.95(2.43) \\
        $n=200$ & 1.82(2.31) & 3.66(4.46) & 7.36(9.15) & 2.30(2.90) & 1.38(1.73) & 1.31(1.62) \\ 
        $n=400$ & 1.23(1.53) & 2.49(3.10) & 5.03(6.26) & 1.74(2.21) & 1.01(1.26) & 0.95(1.19) \\
        \midrule
        $n=100$ & 93.9\% & 92.3\% & 91.6\% & 93.8\% & 95.1\% & 94.4\% \\
        $n=200$ & 92.0\% & 93.3\% & 92.3\% & 95.1\% & 94.5\% & 96.0\% \\ 
        $n=400$ & 95.5\% & 93.8\% & 94.2\% & 93.9\% & 93.8\% & 95.2\% \\
        \bottomrule
    \end{tabular}
    \caption{Simulation results for the Poisson distribution in Study~1. The upper panel shows mean absolute errors (standard errors in parentheses); the lower panel reports empirical coverage of nominal \(95\%\) confidence intervals. All values are multiplied by 100.}
    \label{table:study_1_Poisson}
\end{table}
The overall behavior mirrors that of the Gaussian case. MAE decreases steadily as \(n\) increases, confirming the consistency of the estimators. Compared to the Gaussian setting, the estimation of the correlation parameters \(\balpha\) exhibits slightly larger MAE for small \(n\), which is expected due to the additional complexity of correlated count data.  

Coverage probabilities remain close to the nominal level for the regression coefficients \(\bbeta\), while being slightly under-coverage for \(\balpha\), particularly for small \(n\). As \(n\) increases, coverage for all parameters improves and approaches the desired level.  

These findings suggest that the proposed estimation procedure remains robust when applied to non-Gaussian exponential family responses such as Poisson counts, even when the feasible correlation structure imposes mild additional constraints on data generating process.

\medskip
\noindent{\textbf{Bernoulli data generation.}}
For the Bernoulli distribution, the constraints on valid correlation matrices are even more stringent than in the Poisson case. To handle this, we generate binary responses under a block-structured correlation model, which also serves as an example for the motivating dataset.  

We adopt the generation algorithm proposed by \cite{Emrich:1991}, which is specifically designed for multivariate binary data with a given correlation structure. The covariates \(\bx_{ij}\) remain the same as in the Gaussian and Poisson cases. However, to induce a hierarchical clustering structure, we define  
\[
\bw_{ijk} = (1, \mathbb{I}(v_{ij} = v_{ik}))^\top, 
\]  
where \(v_{ij} \sim \mathrm{Bin}(1, 0.5)\) indicates the subgroup membership of the \(j\)th subject in the \(i\)th group. This setup introduces two correlation levels: a baseline within-group correlation and an additional within-subgroup correlation.  

The linear predictors for the marginal mean and the correlation model are specified as  
\[
\theta_{ij} = \beta_0 + \beta_1 x_{ij,1} + \beta_2 x_{ij,2}, 
\quad 
\gamma_{ijk} = \alpha_0 + \alpha_1 w_{ijk,1}, 
\]  
with the logistic link function \(\mu(x) = \frac{e^x}{1+e^x}\).  
We set the true parameters to  
\[
\bbeta_0 = (1.0, -0.5, 0.5)^\top, \quad 
\balpha_0 = (0.05, 0.15)^\top.
\]  
As in the Poisson case, the Bernoulli distribution has no overdispersion parameter, so \(\phi = 1\).  

It is important to note that the feasible range of correlations between two Bernoulli random variables is constrained and does not span the entire \([-1,1]\) interval. Therefore, if a sampled correlation matrix fails feasibility checks (e.g., it is not positive definite or violates the Fr\'{e}chet bounds), the data generation is repeated until a valid sample is obtained.  

Table~\ref{table:study_1_Bernoulli} summarizes the simulation results. The upper section reports mean absolute errors (MAE) with Monte Carlo standard errors in parentheses, while the lower section shows the empirical coverage probabilities of nominal 95\% confidence intervals.  

\begin{table}[!htbp]
    \centering
    \begin{tabular}{cccccc}
        \toprule
        \multicolumn{6}{c}{\textbf{Bernoulli Distribution}}\\
        Size & $\alpha_0$ & $\alpha_1$ & $\beta_0$ & $\beta_1$ & $\beta_2$\\
        \midrule
        $n=100$ & 3.28(4.06) & 4.74(5.94) & 10.76(13.49) & 8.78(11.14) & 8.87(11.09) \\
        $n=200$ & 2.09(2.59) & 3.13(3.95) & 7.10(8.81) & 6.20(7.69) & 5.93(7.42) \\ 
        $n=400$ & 1.53(1.93) & 2.17(2.71) & 5.12(6.52) & 4.47(5.59) & 4.39(5.41) \\
        \midrule
        $n=100$ & 92.3\% & 93.4\% & 93.7\% & 94.4\% & 94.6\% \\
        $n=200$ & 95.8\% & 93.9\% & 95.7\% & 95.5\% & 95.5\% \\ 
        $n=400$ & 94.1\% & 95.9\% & 94.0\% & 94.1\% & 95.0\% \\ 
        \bottomrule
    \end{tabular}
    \caption{Simulation results for the Bernoulli distribution in Study~1. The upper panel shows mean absolute errors (standard errors in parentheses); the lower panel reports empirical coverage of nominal \(95\%\) confidence intervals. All values are multiplied by 100.}
    \label{table:study_1_Bernoulli}
\end{table}
Compared to the Gaussian and Poisson cases, the Bernoulli setting is more challenging due to the stricter correlation constraints and the binary nature of the responses. As expected, MAE and variability are slightly larger, particularly for the parameters \(\bbeta\) and for small \(n\). Nonetheless, the overall trends remain favorable:  (i) MAE decreases steadily with larger sample sizes, reflecting the asymptotic consistency of the estimators.  
(ii) The coverage probabilities for \(\bbeta\) are close to the nominal level, while those for \(\balpha\) are slightly under-coverage at small \(n\), improving with larger samples.  

These results demonstrate that the proposed method maintains reasonable performance even for Bernoulli responses with complex block correlation structures, thereby reinforcing its robustness across different members of the exponential family.  

\medskip
\noindent{\textbf{Summary across distributions.}} 
Taken together, the three simulation studies show that the proposed method performs reliably across Gaussian, Poisson, and Bernoulli responses. While binary data impose additional constraints and tend to exhibit slightly higher estimation error, the overall trends of decreasing MAE, stable coverage, and improved precision with larger \(n\) are consistent with the theoretical asymptotics established in Section~\ref{sec:theory}.  

\subsection{Comparison with Competitors}  \label{sec:comparison}

We now compare our proposed method with the classical Generalized Estimating Equations (GEE) approach using exchangeable and AR(1) working correlation structures, implemented via the R package \texttt{gee}.  

We consider binary outcomes with a logistic link function \(\mu(x) = \frac{e^x}{1+e^x}\), so that \(a'(x) = \frac{e^x}{1+e^x}\) and \(h(x)=x\). The group size is set as \(m_i \sim \mathrm{Bin}(5,0.8)+2\), while the covariate vector \(\bx_{ij}\) and the linear predictor \(\theta_{ij}\) follow the same setup as in Section~\ref{sec:study1}.  

To ensure a fair comparison, unlike Section \ref{sec:study1}, where the correlation matrices are generated under our proposed model, here we explicitly construct the correlation matrices under different pre-specified structures to mimic clustered and longitudinal data scenarios. This allows us to assess the robustness of each method when the true correlation deviates from the assumed working structure.  

We investigate four cases: two involving hierarchical clustering (Cases~1-2) and two involving longitudinal-type AR(1) correlation (Cases~3-4).  

\medskip
\noindent{\textbf{Case 1: Two-level hierarchical clustering.}}
In this setting, the data exhibit a two-level cluster structure similar to Section \ref{sec:study1}. Each subject within a group belongs to one of two subgroups, indicated by  
\[
u_{ij} \overset{\text{i.i.d.}}{\sim} \mathrm{Bin}(1,0.5).  
\]  
The correlation between subjects \(j\) and \(k\) in group \(i\) is defined as  
\[
\mathrm{corr}(y_{ij}, y_{ik}) = 0.05 \;+\; 0.15 \, \mathbb{I}(u_{ij}=u_{ik}=0) \;+\; 0.2 \, \mathbb{I}(u_{ij}=u_{ik}=1).  
\]  
Thus, subjects in the same subgroup have stronger correlation, and the two subgroups differ slightly in their within-subgroup strength. For our method, we set  
\[
\bw_{ijk} = \big(1, \mathbb{I}(u_{ij} = u_{ik})\big)^\top.  
\]

\medskip
\noindent{\textbf{Case 2: Covariate-modulated clustering.}}
Here we extend Case~1 by incorporating subject-level covariates that further modulate the correlation. Each subject has an additional discrete covariate \(v_{ij} \overset{\text{i.i.d.}}{\sim} \mathrm{Bin}(4,0.5)\). The correlation between two subjects is now  
\[
\begin{aligned}
\mathrm{corr}(y_{ij}, y_{ik}) &= 0.05 \;+\; 0.15 \, \mathbb{I}(u_{ij}=u_{ik}=0) \;+\; 0.2 \, \mathbb{I}(u_{ij}=u_{ik}=1) \\
&\quad -0.05 |x_{ij,1}-x_{ik,1}| -0.05 |x_{ij,2}-x_{ik,2}| -0.05 |v_{ij}-v_{ik}|.  
\end{aligned}
\]  
This structure reflects that larger differences in covariates lead to weaker correlation between subjects. For our method, we enrich the working vector as  
\[
\bw_{ijk} = \big(1, \mathbb{I}(u_{ij}=u_{ik}),\, |x_{ij,1}-x_{ik,1}|,\, |x_{ij,2}-x_{ik,2}|,\, |v_{ij}-v_{ik}|\big)^\top.  
\]

\medskip
\noindent{\textbf{Case 3: AR(1) longitudinal correlation.}}
In this case, the data mimic longitudinal measurements with autoregressive correlation. The correlation between subjects \(j\) and \(k\) in group \(i\) depends on the time-lag \(|j-k|\):  
\[
\mathrm{corr}(y_{ij}, y_{ik}) = 0.4 \, \rho^{|j-k|}, \quad \rho = 0.6.  
\]  
Here, the subject indices \(j\) and \(k\) can be interpreted as observation times, with correlation decaying exponentially as the time gap increases. For our method, the working covariate vector is simply  
\[
\bw_{ijk} = \big(1, |j-k|\big)^\top.  
\]

\medskip
\noindent{\textbf{Case 4: Covariate-modulated AR(1) correlation.}}
Finally, we combine the AR(1) temporal decay with the covariate effects as in Case~2. Thus,  
\[
\begin{aligned}
\mathrm{corr}(y_{ij}, y_{ik}) &= 0.4 \, \rho^{|j-k|} \;-\;0.05 |x_{ij,1}-x_{ik,1}| -0.05 |x_{ij,2}-x_{ik,2}| -0.05 |v_{ij}-v_{ik}|, \\
\rho &= 0.6.
\end{aligned}
\]  
This represents a realistic longitudinal scenario where correlation decays over time but is also weakened when subjects differ substantially in key covariates. For our method, we adopt  
\[
\bw_{ijk} = \big(1, |j-k|,\, |x_{ij,1}-x_{ik,1}|,\, |x_{ij,2}-x_{ik,2}|,\, |v_{ij}-v_{ik}|\big)^\top.  
\]

\medskip
\noindent{\textbf{Remarks.}}
These four scenarios cover a range of realistic data structures: (i) pure clustering, (ii) clustering with covariate modulation, (iii) pure temporal AR(1) dependence, and (iv) AR(1) dependence combined with covariate effects.  

For each case, we will evaluate estimation accuracy and inference performance of our proposed method versus the GEE with exchangeable and AR(1) working correlation assumptions. This comparison highlights the advantages of explicitly modeling flexible pairwise correlation structures versus using fixed-form working correlations.  

For performance comparison, we evaluate each method using two error measures that quantify how well the mean and covariance structures are recovered:  
\[
\mathrm{MMD} = \frac{1}{n} \sum_{i=1}^n \|\hat{\bmu}_i - \bmu_{0i}\|,  
\quad
\mathrm{MCD} = \frac{1}{n} \sum_{i=1}^n \|\hat{\bSigma}_i - \bSigma_{0i}\|,
\]
where \(\|\cdot\|\) denotes the \(\ell_2\)-norm for the mean error and the Frobenius norm for the covariance error.  Here MMD (Mean Mean Deviation) measures the average discrepancy between the estimated marginal means \(\hat{\bmu}_i\) and their true values \(\bmu_{0i}\), while  MCD (Mean Covariance Deviation) measures the average discrepancy between the estimated covariance matrices \(\hat{\bSigma}_i\) and the true \(\bSigma_{0i}\).  A smaller value indicates a more accurate recovery of the underlying structure.  

We fix the true regression coefficients as \(\bbeta_0 = (1.0,\,-0.5,\,0.5)^\top\), while the correlation structure varies across the four cases described earlier. We consider three sample sizes, \(n = 100, 200, 400\), and each scenario is repeated 1000 times.  

For each case, we report the mean MMD and mean MCD, along with their corresponding standard errors (all values are multiplied by 100 for readability). In each setting, the smallest mean MMD or MCD among all competing methods is \textbf{underlined}.  The detailed results are presented in Table~\ref{table:study_2}.  

\begin{table}[!htbp]
    \centering
    \resizebox*{1\textwidth}{!}{
        \begin{tabular}{cccccccc}
            \toprule
                \multirow{2}{*}{\textbf{Case}} & 
                \multirow{2}{*}{\textbf{Size}} & 
                \multicolumn{2}{c}{\textbf{Proposed Method}} & 
                \multicolumn{2}{c}{\textbf{GEE (Exch)}} & 
                \multicolumn{2}{c}{\textbf{GEE (AR-1)}} \\
                \cmidrule(lr){3-4} \cmidrule(lr){5-6} \cmidrule(lr){7-8} &
                & \textbf{MMD} & \textbf{MCD} & \textbf{MMD} & \textbf{MCD} & \textbf{MMD} & \textbf{MCD} \\
            \midrule
                \multirow{3}{*}{Case 1} & $n=100$ & \underline{7.27}(3.11) & \underline{6.53}(2.26) & 7.34(3.15) & 10.74(1.22) & 7.48(3.17) & 15.07(0.80) \\
                 & $n=200$ & \underline{5.22}(2.36) & \underline{4.96}(1.57) & 5.29(2.39) & 10.10(0.77) & 5.37(2.42) & 14.72(0.50) \\
                 & $n=400$ & \underline{3.71}(1.58) & \underline{3.89}(1.02) & 3.77(1.59) & 9.73(0.41) & 3.81(1.60) & 14.53(0.33) \\
            \midrule
                \multirow{3}{*}{Case 2} & $n=100$ & \underline{6.72}(2.86) & \underline{7.18}(1.89) & 6.87(2.90) & 12.62(0.59) & 6.87(2.91) & 12.76(0.62) \\
                 & $n=200$ & \underline{4.72}(2.00) & \underline{5.27}(1.27) & 4.84(2.04) & 12.30(0.32) & 4.85(2.04) & 12.42(0.34) \\
                 & $n=400$ & \underline{3.42}(1.45) & \underline{4.06}(0.86) & 3.52(1.51) & 12.14(0.20) & 3.52(1.52) & 12.24(0.20) \\
            \midrule
                \multirow{3}{*}{Case 3} & $n=100$ & \underline{7.32}(3.24) & \underline{6.80}(2.27) & 7.41(3.25) & 9.89(1.53) & 7.33(3.24) & 7.74(1.64) \\
                 & $n=200$ & \underline{5.21}(2.35) & \underline{5.14}(1.41) & 5.23(2.37) & 8.95(0.80) & 5.22(2.34) & 6.94(1.05) \\
                 & $n=400$ & \underline{3.74}(1.71) & \underline{4.16}(0.86) & 3.77(1.73) & 8.51(0.42) & 3.74(1.70) & 6.55(0.72) \\
            \midrule
                \multirow{3}{*}{Case 4} & $n=100$ & \underline{6.81}(2.86) & \underline{7.19}(1.71) & 6.94(2.94) & 11.75(0.56) & 6.96(2.93) & 10.78(0.75) \\
                 & $n=200$ & \underline{4.89}(2.11) & \underline{5.42}(1.14) & 4.99(2.15) & 11.43(0.34) & 5.00(2.15) & 10.38(0.47)\\
                 & $n=400$ & \underline{3.46}(1.45) & \underline{4.24}(0.70) & 3.51(1.49) & 11.22(0.20) & 3.51(1.49) & 10.11(0.25) \\
            \bottomrule
        \end{tabular}
    }
    \caption{Simulation results comparing our method and GEE.}
    \label{table:study_2}
\end{table}
It is notable that the proposed method consistently outperforms the two GEE competitors across all cases and sample sizes, both in terms of estimating the marginal means and recovering the covariance matrix. In particular, the advantage is most pronounced for covariance estimation, where incorporating covariate-driven correlations allows the proposed model to capture complex dependence structures that traditional working correlations (exchangeable or AR-1) fail to approximate well. Moreover, when influential covariates are present, our method also achieves clear gains in mean estimation accuracy, highlighting its adaptability to heterogeneous correlation patterns.

Beyond improved estimation accuracy, an important strength of the proposed method is that it yields direct estimates and valid inference for the correlation parameters 
$\balpha$. This enables a deeper understanding of how covariates affect within-group correlations, an insight that is generally unattainable under standard GEE frameworks with fixed working structures. Overall, these results validate the effectiveness and interpretability of the proposed joint mean-covariance estimation approach.

\section{Cross-validation Setup}

We implement a 5-fold (stratified) cross-validation procedure repeated 15 times and compute two commonly used metrics: the Brier Score \citep{Brier:1950} and the Log Loss, defined as
\begin{align*}
    \text{Brier Score} &= \frac{1}{N} \sum_{i=1}^n \sum_{j=1}^{m_i} (y_{ij} - \hat{p}_{ij})^2,\\
    \text{Log Loss} &= -\frac{1}{N} \sum_{i=1}^n \sum_{j=1}^{m_i} \left[ y_{ij}\log \hat{p}_{ij} + (1-y_{ij})\log(1-\hat{p}_{ij}) \right],
\end{align*}
where $n$ is the number of clusters, $m_i$ the size of cluster $i$, $N$ the total number of observations, $y_{ij}$ the observed outcome, and $\hat{p}_{ij}$ the predicted probability.  

The Brier Score captures overall predictive accuracy as the mean squared error between predicted probabilities and observed outcomes, while the Log Loss emphasizes the quality of probabilistic predictions by heavily penalizing overconfident misclassifications.  
For both metrics, smaller values indicate better predictive performance.

For each repetition, the performance is averaged over the 5 folds:
\begin{gather*}
    \bar{s}_{r,1} = \frac{1}{5} \sum_{k=1}^{5} \text{Brier Score}_{r,k}, \quad
    \bar{s}_{r,2} = \frac{1}{5} \sum_{k=1}^{5} \text{Log Loss}_{r,k},
\end{gather*}
and the overall cross-validated score is computed as
\begin{equation*}
    \text{CV}_l = \frac{1}{15} \sum_{r=1}^{15} \bar{s}_{r,l}, \quad l = 1,2,
\end{equation*}
where $r$ indexes repetitions and $k$ indexes folds.

\section{Additional Analysis of Modern Prenatal Care Data}

\subsection{Analysis of the Full Correlation Models}

We extend the initial baseline correlation model to incorporate socio-demographic and household-level covariates relevant to prenatal care choices:
\begin{align}\label{eq:corr_family_full}
    \gamma_{ijk} = & \ \alpha_0 
    + \alpha_1\, w_{ijk,1} 
    + \sum_{l=0}^2 \alpha_{2+l}\, w_{ijk,2+l} 
    + \sum_{l=0}^4 \alpha_{5+l}\, w_{ijk,5+l} 
    + \alpha_{10}\, w_{ijk,10} 
    + \alpha_{11}\, w_{ijk,11} \notag\\
    & + \alpha_{12}\, w_{ijk,12}
    + \alpha_{13}\, w_{ijk,13},
\end{align}
where \(i\) index communities and \(j,k\) index records,.  

The additional covariates are defined as:
\begin{itemize}
    \item \(w_{ijk,2+l} = \mathbb{I}(\text{indig}_{ij} = \text{indig}_{ik} = l)\), for \(l = 0,1,2\), indicating shared ethnic-linguistic identity;
    \item \(w_{ijk,5+l} = \mathbb{I}(\text{husEmpl}_{ij} = \text{husEmpl}_{ik} = l)\), for \(l = 0,1,2,3,4\), indicating shared employment status of the husband;
    \item \(w_{ijk,10} = |\text{momEd}_{ij} - \text{momEd}_{ik}|\), the absolute difference in the mother’s education level;
    \item \(w_{ijk,11} = |\text{husEd}_{ij} - \text{husEd}_{ik}|\), the absolute difference in the husband’s education level;
    \item \(w_{ijk,12} = |\text{toilet}_{ij} - \text{toilet}_{ik}|\), the absolute difference in modern toilet usage;
    \item \(w_{ijk,13} = |\text{TV}_{ij} - \text{TV}_{ik}|\), the absolute difference in TV viewing frequency.
\end{itemize}

To further assess the contribution of generic familial clustering, we also estimate a reduced correlation model that excludes the family-level indicator \(w_{ijk,1}\).  
This specification allows the effects of specific family-level covariates to be evaluated in isolation:
\begin{align}\label{eq:corr_family_full_reduced}
    \gamma_{ijk} = & \ \alpha_0  
    + \sum_{l=0}^2 \alpha_{2+l}\, w_{ijk,2+l} 
    + \sum_{l=0}^4 \alpha_{5+l}\, w_{ijk,5+l}
    + \alpha_{10}\, w_{ijk,10} 
    + \alpha_{11}\, w_{ijk,11} \notag\\
    & + \alpha_{12}\, w_{ijk,12}
    + \alpha_{13}\, w_{ijk,13},
\end{align}
where the notation follows that of~\eqref{eq:corr_family_full}.  

The parameter estimates and corresponding $p$-values for models~\eqref{eq:corr_family_full} and~\eqref{eq:corr_family_full_reduced} are summarized in Table~\ref{table:corr_family_full}, facilitating comparison of the influence of generic versus covariate-specific correlation components.

\begin{table}[!htbp]
    \centering
    \resizebox*{1\textwidth}{!}{
        \sisetup{
            table-format = -1.4,
            table-space-text-post = {***},
            table-align-text-post = false,
            table-number-alignment = center,
        }
        \begin{tabular}{@{} c 
            S[table-format=-1.4, table-column-width=3cm]
            S[table-format=1.4, table-column-width=2cm]
            S[table-format=-1.4, table-column-width=3cm]
            S[table-format=1.4, table-column-width=2cm]
            @{}}
            \toprule
                \multirow{2}{*}{\textbf{Variable}} & 
                \multicolumn{2}{c}{\textbf{Model \eqref{eq:corr_family_full}}} & 
                \multicolumn{2}{c}{\textbf{Model \eqref{eq:corr_family_full_reduced}}} \\
            \cmidrule(lr){2-3} \cmidrule(lr){4-5}
                & \textbf{Estimate} & \textbf{P-value} & \textbf{Estimate} & \textbf{P-value} \\
            \midrule
                Intercept ($\alpha_0$) & 0.0471\textbf{.} & 0.0937 & 0.0739* & 0.0164 \\
                mom ($\alpha_1$) & 0.8266*** & 0.0000 & \text{--} & \text{--} \\
                \hdashline
                Ladino ($\alpha_2$) & -0.0187 & 0.5628 & 0.0203 & 0.4893 \\
                Indigenous non-Spanish ($\alpha_3$) & -0.0095 & 0.8236 & -0.0253 & 0.5191 \\
                Indigenous Spanish ($\alpha_4$) & -0.0617\textbf{.} & 0.0799 & -0.0498 & 0.1341 \\
                \hdashline
                Unskilled ($\alpha_5$) & 0.3592 & 0.3001 & 0.8454*** & 0.0000 \\
                Professional ($\alpha_6$) & 0.0756 & 0.7874 & 0.2930** & 0.0013 \\
                Agri-self ($\alpha_7$) & 0.0331\textbf{.} & 0.0823 & 0.0438* & 0.0102 \\
                Agri-empl ($\alpha_8$) & 0.0115 & 0.6601 & 0.0541*** & 0.0006 \\
                Skilled ($\alpha_9$) & 0.0629 & 0.3018 & 0.1842*** & 0.0000 \\
                \hdashline
                momEd ($\alpha_{10}$) & 0.0066 & 0.6704 & -0.0279* & 0.0413 \\
                husEd ($\alpha_{11}$) & 0.0041 & 0.7744 & -0.0352** & 0.0023 \\
                \hdashline
                toilet ($\alpha_{12}$) & 0.0073 & 0.8331 & -0.0413 & 0.2276 \\
                TV ($\alpha_{13}$) & -0.0004 & 0.9809 & -0.0186 & 0.2135 \\
            \bottomrule
        \end{tabular}
    }

    \medskip
    \flushleft{
        \small 
        Notes: \textbf{.} for $\text{p} < 0.1$, * for $\text{p} < 0.05$, ** for $\text{p} < 0.01$ and *** for $\text{p} < 0.001$.
    }
    
    \caption{Parameter estimates and corresponding p-values for \eqref{eq:corr_family_full} and \eqref{eq:corr_family_full_reduced}.}
    \label{table:corr_family_full}
\end{table}
In Model~\eqref{eq:corr_family_full}, the hierarchical clustering effects \(\hat{\alpha}_0\) and \(\hat{\alpha}_1\) remain significant, while most additional family-level covariates are insignificant with near-zero effects, likely because the strong generic familial clustering \(w_{ijk,1}\) captures most of their potential variation.
Similarly, the two exceptions are \(\hat{\alpha}_4\) (Indigenous Spanish ethnic-linguistic group) and \(\hat{\alpha}_7\) (husband agri-self employment), both marginally significant (p $<0.1$), indicating additional dependence beyond generic family- and community-level effects. 

For Model~\eqref{eq:corr_family_full_reduced}, which omits the generic familial clustering term \(w_{ijk,1}\), several notable patterns emerge:

\begin{enumerate}
    \item \textbf{Effect of removing \(w_{ijk,1}\):} Without the generic familial clustering term, parameters for the family-level covariates generally have larger magnitudes and lower p-values compared to Model~\eqref{eq:corr_family_full}. This confirms that the model effectively partitions overall familial clustering into more detailed subcomponents.
    
    \item \textbf{Ethnic-linguistic identity:} The estimate for \(\hat{\alpha}_2\) (Ladino) is positive, whereas \(\hat{\alpha}_3\) and \(\hat{\alpha}_4\) (indigenous groups) are negative. While these signs suggest potential positive clustering within Ladino families and negative clustering within indigenous populations, none of these effects are statistically significant, indicating that ethnic-linguistic clustering is weak in this context.
    
    \item \textbf{Husband’s employment status:} All five employment-related parameters are significantly positive, implying that mothers whose partners share the same employment category tend to have similar prenatal care patterns. The smallest effect is observed for the agri-self category (\(\hat{\alpha}_7\)), while the largest is for the unskilled category (\(\hat{\alpha}_5\)), reflecting the strongest clustering.
    
    \item \textbf{Parental education differences:} Both \(\hat{\alpha}_{10}\) (maternal education) and \(\hat{\alpha}_{11}\) (paternal education) are significantly negative, indicating that smaller differences in parental education are associated with higher similarity in prenatal care, even after accounting for other covariates.
    
    \item \textbf{Household amenities:} The coefficients \(\hat{\alpha}_{12}\) (toilet usage) and \(\hat{\alpha}_{13}\) (TV viewing) are negative but not statistically significant, suggesting that these amenities have little impact on the correlation structure.
\end{enumerate}

These results reveal insights often obscured by conventional models that capture only generic clustering. While community- and family-level dependence affects prenatal care choices, the underlying social and structural drivers typically remain hidden. Model~\eqref{eq:corr_family_full_reduced} uncovers nuanced correlation patterns linked to social, ethnic, socio-demographic, and economic factors, highlighting mechanisms behind observed clustering and suggesting targets for health interventions or further investigation.

To examine individual-level factors that may influence correlations in prenatal care choices more closely, we extend Model~\eqref{eq:corr_family_full} by incorporating pairwise differences in selected individual covariates:
\begin{align}\label{eq:corr_full}
    \gamma_{ijk} = \ & \alpha_0 + \alpha_1 w_{ijk,1} 
    + \sum_{l=0}^2 \alpha_{2+l} w_{ijk,2+l} 
    + \sum_{l=0}^4 \alpha_{5+l} w_{ijk,5+l} 
    + \alpha_{10} w_{ijk,10} + \alpha_{11} w_{ijk,11} \notag \\
    & + \alpha_{12} w_{ijk,12} + \alpha_{13} w_{ijk,13} 
    + \alpha_{14} w_{ijk,14} + \alpha_{15} w_{ijk,15} + \alpha_{16} w_{ijk,16},
\end{align}
where $i,j,k$ follow the same indexing as in~\eqref{eq:corr_family_full}.  
The additional covariates are defined as
\[
w_{ijk,14} = |\mathrm{motherAge}_{ij} - \mathrm{motherAge}_{ik}|,\quad
w_{ijk,15} = |\mathrm{childAge}_{ij} - \mathrm{childAge}_{ik}|,\]
\[w_{ijk,16} = |\mathrm{birthOrd}_{ij} - \mathrm{birthOrd}_{ik}|,
\]
representing pairwise differences in maternal age, child age, and birth order, respectively.  
All other covariates remain as specified in~\eqref{eq:corr_family_full}.  
The corresponding parameter estimates and p-values are summarized in Table~\ref{table:corr_full}.

\begin{table}[!h]
    \centering
    \sisetup{
        table-format = -1.4,
        table-space-text-post = {***},
        table-align-text-post = false,
        table-number-alignment = center,
    }
    \begin{tabular}{@{} c 
        S[table-format=-1.4, table-column-width=6cm]
        S[table-format=1.4, table-column-width=3cm]
        @{}}
    \toprule
    \textbf{Variable} & \textbf{Estimate} & \textbf{P-value}\\
    \midrule
    Intercept ($\alpha_0$) & 0.0683* & 0.0283 \\
    mom ($\alpha_1$) & 0.8269*** & 0.0000 \\
    \hdashline
    Ladino ($\alpha_2$) & -0.0172 & 0.6092 \\
    Indigenous non-Spanish ($\alpha_3$) & -0.0034 & 0.9389 \\
    Indigenous Spanish ($\alpha_4$) & -0.0648\textbf{.} & 0.0792 \\
    \hdashline
    Unskilled ($\alpha_5$) & 0.3768 & 0.2653 \\
    Professional ($\alpha_6$) & 0.0944 & 0.7185 \\
    Agri-self ($\alpha_7$) & 0.0293 & 0.1467 \\
    Agri-empl ($\alpha_8$) & 0.0166 & 0.5786 \\
    Skilled ($\alpha_9$) & 0.0578 & 0.3209 \\
    \hdashline
    momEd ($\alpha_{10}$) & 0.0095 & 0.5445 \\
    husEd ($\alpha_{11}$) & 0.0027 & 0.8541 \\
    \hdashline
    toilet ($\alpha_{12}$) & 0.0061 & 0.8722 \\
    TV ($\alpha_{13}$) &  0.0021 & 0.8981 \\
    \hdashline
    motherAge ($\alpha_{14}$) & 0.0072 & 0.6238 \\
    childAge ($\alpha_{15}$) & -0.0251 & 0.1906 \\
    birthOrd ($\alpha_{16}$) & -0.0130 & 0.3412 \\
    \bottomrule
    \end{tabular}

    \medskip
    \flushleft{
        \small 
        Notes: \textbf{.} for $\text{p} < 0.1$, * for $\text{p} < 0.05$, ** for $\text{p} < 0.01$ and *** for $\text{p} < 0.001$.
    }
    
    \caption{Parameter estimates and corresponding p-values for \eqref{eq:corr_full}.}
    \label{table:corr_full}
\end{table}
Analysis of Model~\eqref{eq:corr_full} yields results broadly consistent with our earlier findings.  
Although none of the additional individual-level covariates reach statistical significance, both $\hat{\alpha}_{15}$ and $\hat{\alpha}_{16}$ are estimated to be negative. This suggests a gradual weakening of shared prenatal care patterns over time and reduced similarity in care choices among children with larger age differences in larger families. Such patterns may reflect temporal changes in access, preferences, or health policies, as well as generational differences within multi-child households. However, given the limited data for these specific covariates, these observations should be interpreted cautiously, and further investigation is warranted before drawing definitive conclusions.  

Overall, our comprehensive analysis confirms that prenatal care choices exhibit strong dependence at both the family and community levels, with familial clustering being the dominant factor. By extending the correlation model to incorporate socio-demographic, ethnic, and economic covariates, our approach reveals additional influences---such as ethnic-linguistic identity, husband’s employment type, and parental education similarity---that standard clustering models may obscure.  
These findings offer a more nuanced understanding of the social and cultural mechanisms shaping care patterns and provide actionable insights for targeted public health interventions and policy design.  
The flexibility of our model, which allows general clustering to be decomposed into interpretable, covariate-driven components, makes it a valuable tool for applied analyses aiming to understand the drivers of correlated behaviors.
\subsection{Comparisons with GLMM}

We assess the performance of the GLMM as a main competing approach in this context. 
In addition to the baseline community- and family-level effects, we explore extensions incorporating additional hierarchical random effects. 
The GLMM is initially fitted using the \texttt{lme4} package \citep{Bates:2015}. 
For more complex specifications, where convergence issues arise due to model complexity, we employ the \texttt{glmmTMB} package, which is generally more robust for large models with multiple random effects. 
The corresponding results are presented in Tables~\ref{table:corr_full_glmm_1} and~\ref{table:corr_full_glmm_2}.

To assess model adequacy, likelihood ratio tests (LRTs) are commonly applied in GLMMs. 
However, in this setting, such tests often become inapplicable or unreliable because optimization may fail to converge for certain reduced models. 
Furthermore, when the null hypothesis places the parameter on the boundary of the parameter space---and the corresponding estimates are close to this boundary---the standard LRT is known to be unreliable and requires specialized boundary-adjusted inference \citep{Self:1987}. 
Accordingly, ``NA'' is reported to indicate cases where significance testing is not available. 
Given that the estimated standard deviations cluster into two groups---one with relatively large values (Family, husEmpl, motherAge, childAge, birthOrd) and another approaching zero (Community, indig, momEd, husEd, toilet, TV)---we report significance levels using the standard LRT for the former group and the adjusted LRT for the latter.

\begin{table}[!h]
    \centering
    \begin{tabular}{cccccccccccc}
        \toprule
        \textbf{Estimate} & {Community} & {Family} & {indig} & {husEmpl} & {momEd} & {husEd} \\
        \midrule
        \textbf{Std.Dev.} & $0.0084$ & $24.2176$ & $0.0119$ & $1.0357$ & $0.0046$ & $0.0020$ \\
        \textbf{$p$-Value} & NA & NA & $0.4950$ & NA & NA & NA \\
        \bottomrule
    \end{tabular}
    
  \caption{Estimates of the random effects by the GLMM.}
  \label{table:corr_full_glmm_1}
\end{table}

\begin{table}[!h]
    \centering
    \begin{tabular}{cccccccccccc}
        \toprule
        \textbf{Estimate} & {toilet} & {TV} & {motherAge} & {childAge} & {birthOrd} \\
        \midrule
        \textbf{Std.Dev.} & $0.0127$ & $0.0019$ & $1.0726$ & $3.3576$ & $3.8209$ \\
        \textbf{$p$-Value} & NA & $1.0000$ & $0.8731$ & NA & $1.0000$ \\
        \bottomrule
    \end{tabular}
    
  \caption{Estimates of the random effects by the GLMM.}
  \label{table:corr_full_glmm_2}
\end{table}

While inference with the GLMM exhibits substantial limitations---both in the difficulty of conducting reliable statistical inference and in its inability to disentangle the contributions of individual covariates when exploring correlation structures---the results obtained when the GLMM is applicable are broadly consistent with those from our approach. In particular, both methods identify a significant family-level clustering effect and a moderate community-level clustering effect.

We conclude that our proposed approach offers several key advantages over standard GLMMs: it accommodates complex correlation structures with greater flexibility, facilitates straightforward testing of parameter significance, enhances interpretability of how specific covariates influence the correlation, and enables fine-grained analysis---such as distinguishing the effects of different categories of fathers’ employment status---which is often infeasible with conventional GLMM methods.

\section{Analysis of Grouse Ticks Data}\label{sec:real_3}

We apply our method to analyze a clustered dataset of count data --- the grouse ticks dataset from \cite{Elston:2001}, which documents a study of tick parasitism in red grouse chicks in Scotland from 1995 to 1997. The dataset is unbalanced, which contains 403 observations of tick infestation counts on individual chicks, nested within 118 distinct broods that are further grouped into 63 geographical locations. Thus, the dataset exhibit a two-level hierarchical clustering structure similar to the modern prenatal care dataset in Section~3 of the main part. Details of the dataset are shown in Table~\ref{table:data_3_desc} with the response variable underlined.

\begin{table}[!h]
    \centering
    \begin{tabular}{cc}
        \toprule
        \textbf{Variable}  & \textbf{Description} \\
        \midrule
        Index & Unique identifier for each chick \\
        \underline{Ticks} & Count of ticks found on each chick: $[0, 85]$ \\
        Brood & Unique identifier for each brood \\
        Location & Unique identifier for each location \\
        cHeight & Centered height above sea level for each location: $[-59.24,70.76]$ \\
        Year & Year for each observation: $0=1995,1=1996,2=1997$ \\
        \bottomrule
    \end{tabular}
    \caption{Description of variables of the grouse ticks dataset.}
    \label{table:data_3_desc}
\end{table}

It is worth mentioning that the dataset exhibits substantial over-dispersion ($\mathrm{Variance}=172.6615$ and $\mathrm{Mean}=6.3697$) and zero-inflated feature that violate the standard Poisson assumption; see Figure~\ref{fig:overdispersion}. 
\begin{figure}[!h]
    \centering
    \includegraphics[width=0.8\textwidth]{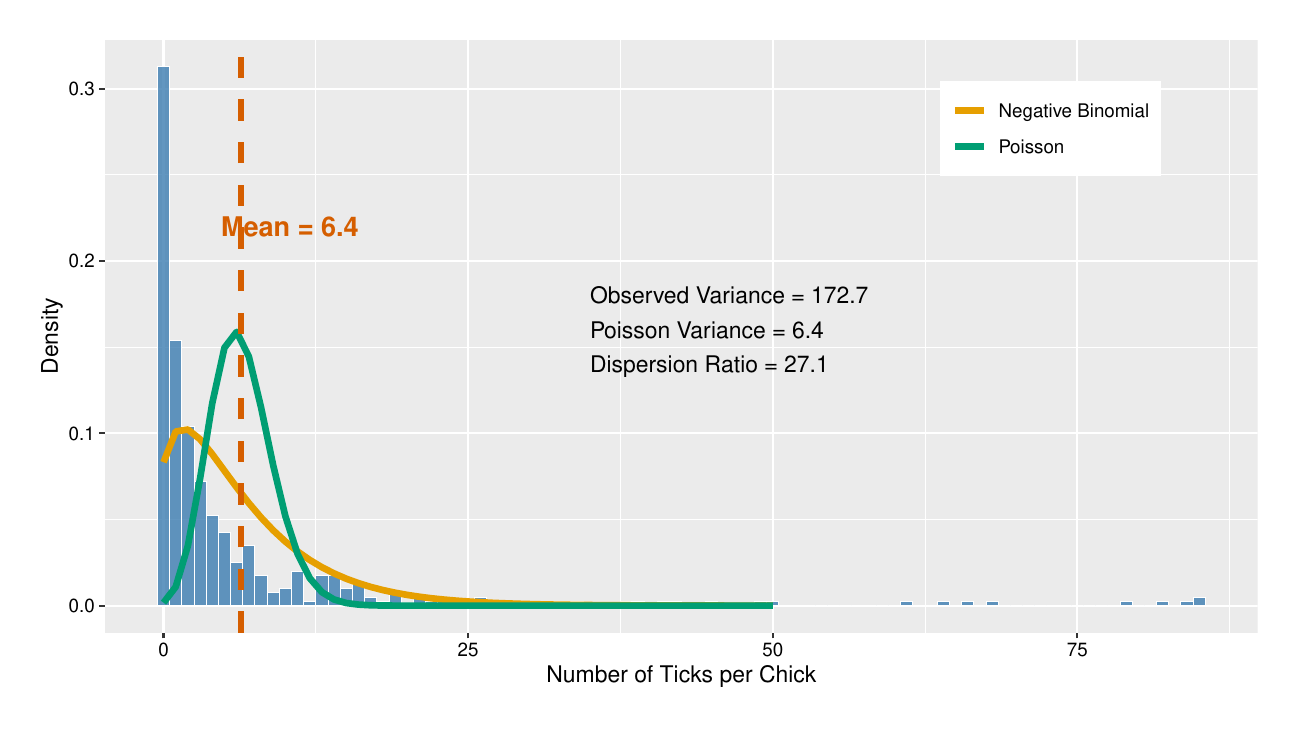}
    \caption{Over-dispersion and excess zeros in the grouse ticks dataset.}
    \label{fig:overdispersion}
\end{figure}

Let $y_{ij}$ denote the response for the $j$th chick at the $i$th location, with $\mu_{ij} = \mathbb{E}(y_{ij})$.  
We specify the model as  
\begin{align*}
    \log(\mu_{ij}) &= \beta_0 + \beta_1 \text{cHeight}_{ij} + \sum_{k=1}^2 \beta_{1+k} \, \mathbb{I}(\text{Year}_{ij} = k), \\
    \gamma_{ijk} &= \alpha_0 + \alpha_1 w_{ijk,1},
\end{align*}
where $i$ indexes locations, and $j$ and $k$ index chicks.  
The indicator $w_{ijk,1}$ equals 1 if chicks $j$ and $k$ are from the same brood and 0 otherwise.  
Here, $\alpha_0$ captures the within-location correlation, while $\alpha_1$ measures the additional correlation among chicks from the same brood.


To compare with the GLMM, we fit a zero-inflated negative binomial (ZINB) mixed-effects model using the R package \texttt{glmmTMB}.  
The model specification is:  
\begin{align*}
    y_{ijk} & \sim \mathrm{ZINB}(\mu_{ijk}, \theta, \pi_{ijk}), \\
    \log(\mu_{ijk}) &= \beta_0 + \beta_1 \text{cHeight}_{ijk} + \sum_{l=1}^2 \beta_{1+l} \, \mathbb{I}(\text{Year}_{ijk} = l) + u_i + v_{ij}, \\
    \mathrm{logit}(\pi_{ijk}) &= \gamma_0 + \gamma_1 \text{cHeight}_{ijk} + \sum_{l=1}^2 \gamma_{1+l} \, \mathbb{I}(\text{Year}_{ijk} = l),
\end{align*}
where $i$, $j$, and $k$ index locations, broods, and chicks, respectively;  
$\mu_{ijk} = \mathbb{E}(y_{ijk})$; $\theta$ is the dispersion parameter; and $\pi_{ijk}$ is the conditional probability of a structural zero.  
The random effects $u_i \sim \mathrm{N}(0, \sigma_u^2)$ and $v_{ij} \sim \mathrm{N}(0, \sigma_v^2)$ capture within-location and within-brood variation, respectively.  
For comparison, we also fit a Poisson mixed-effects model with the same specification but replacing the negative binomial distribution with the Poisson distribution.

The estimating results for our proposed method and the GLMM are reported in Table~\ref{table:data3_estimate}, together with their corresponding p-values.
\begin{table}[h]
    \centering
    \sisetup{
        table-format = -1.4,  
        table-space-text-post = {***}, 
        table-align-text-post = false, 
        table-number-alignment = center,
    }
    \resizebox*{1\textwidth}{!}{
    \begin{tabular}{@{} c 
        S[table-format=-1.4, table-column-width=2.5cm] 
        S[table-format=1.4, table-column-width=2.0cm]   
        S[table-format=-1.4, table-column-width=2.5cm] 
        S[table-format=1.4, table-column-width=2.0cm]   
        S[table-format=-1.4, table-column-width=2.5cm] 
        S[table-format=1.4, table-column-width=2.0cm]   
        @{}}
    \toprule
    \multirow{2}{*}{\textbf{Variable}} & 
    \multicolumn{2}{c}{\textbf{Proposed Method}} & 
    \multicolumn{2}{c}{\textbf{GLMM (ZINB)}} & 
    \multicolumn{2}{c}{\textbf{GLMM (Poisson)}} \\
    \cmidrule(lr){2-3} \cmidrule(lr){4-5} \cmidrule(lr){6-7}
    & \textbf{Estimate} & \textbf{P-value} & \textbf{Estimate} & \textbf{P-value} & \textbf{Estimate} & \textbf{P-value} \\
    \midrule
    Intercept ($\beta_0$) & 1.4693*** & 0.0000 & 0.7095** & 0.0017 & 0.4669* & 0.0150 \\
    cHeight ($\beta_1$) & -0.0231*** & 0.0000 & -0.0229*** & 0.0001 & -0.0235*** & 0.0000\\
    Year-1996 ($\beta_2$) & 0.5184*** & 0.0000 & 0.9848*** & 0.0000 & 1.1656*** & 0.0000\\
    Year-1997 ($\beta_3$) & -1.6823*** & 0.0000 & -1.1659*** & 0.0000 & -0.9779*** & 0.0001 \\
    \hdashline
    Location ($\alpha_0/\sigma_u$) & 0.0098 & 0.1713 & 0.5697 & 0.1135 & 0.5741 & 0.1381 \\
    Brood ($\alpha_1/\sigma_v$) & 0.2783*** & 0.0000 & 0.6748*** & 0.0000 & 0.7697*** & 0.0000 \\
    \bottomrule
    \end{tabular}
    }

    \medskip
    \flushleft{
        \small
        Notes: \textbf{.} for $\text{p} < 0.1$, * for $\text{p} < 0.05$, ** for $\text{p} < 0.01$ and *** for $\text{p} < 0.001$.
    }
    
    \caption{Parameter estimates and corresponding p-values of our method and the GLMM.}
    \label{table:data3_estimate}
\end{table}

From Tables~\ref{table:data3_estimate}, all three models lead to consistent conclusions:  
(i) the covariates \textit{cHeight} and \textit{Year} have the same sign and significance pattern across models; and  
(ii) there is a significant clustering effect at the brood level, whereas the location-level effect is insignificant.  

We also compare predictive accuracy using the mean absolute error (MAE):  
\begin{equation*}
    \mathrm{MAE} = \frac{1}{N} \sum_{i=1}^n \sum_{j=1}^{m_i} \left| y_{ij} - \hat{y}_{ij} \right|,
\end{equation*}
where $i$ indexes locations, $j$ indexes chicks, $N$ is the total number of observations, $y_{ij}$ is the observed count, and $\hat{y}_{ij}$ is the predicted mean count.  
For the GLMMs, $\hat{y}_{ij}$ corresponds to $\hat{\mu}_{ijk}$ in model notation, obtained using the \texttt{predict} function for \texttt{glmmTMB}.  

The results in Table~\ref{table:data3_compare} show that, despite severe over-dispersion and zero inflation, our proposed method achieves the smallest MAE, with the GLMM (Poisson) performing worst as expected.  
This demonstrates both the superior predictive accuracy and robustness of our approach in challenging applied settings.  
From an applied perspective, the ability to capture fine-scale brood-level clustering while maintaining predictive accuracy is valuable for ecological and wildlife management studies, where understanding group-specific risk factors (e.g., tick infestation in grouse chicks) is crucial for designing effective monitoring and intervention strategies.

\begin{table}[!h]
    \centering
      \begin{tabular}{cccc}
        \toprule
        \textbf{Criteria} & \textbf{Proposed Model} & \textbf{GLMM (ZINB)} & \textbf{GLMM (Poisson)} \\
        \midrule
        \textbf{MAE} & \underline{5.5691} & 8.4828 & 8.9188 \\
        \bottomrule
      \end{tabular}
  \caption{Model comparison for the grouse ticks dataset.}
  \label{table:data3_compare}
\end{table}

\end{appendix}

\end{document}